\documentclass[preprint,5p,authoryear]{elsarticle}
\usepackage{indentfirst}
\usepackage{circuitikz}
\usepackage{amsmath,amssymb}
\usepackage{graphicx,subfigure,epstopdf}
\usepackage{epic}
\usepackage{subfigure}
\usepackage{graphicx}
\usepackage{epsfig}
\usepackage{arydshln}
\usepackage{enumerate}
\usepackage[english]{babel}
\usepackage{fancyhdr}
\usepackage{lastpage}
\usepackage{color}
\usepackage{circuitikz}
\usepackage{tikz}
\usepackage{dsfont}
\usepackage{caption}
\usepackage{changes}
\usepackage{booktabs}
\usepackage{threeparttable}
\usepackage{animate}
\usepackage{natbib}
\usepackage{caption}
\usepackage[hyperfootnotes=false]{hyperref}
\hypersetup{
colorlinks=true,
linkcolor=black,
citecolor=cyan
}
\usepackage{multirow}
\usepackage{bigdelim}
\usepackage{algorithm}
\usepackage{algpseudocode}
\usepackage[bottom]{footmisc}
\usetikzlibrary{arrows,shapes,chains}
\setcitestyle{longnamesfirst,round,authoryear,semicolon,sort&compress}
\usetikzlibrary{calc}
\usepackage{multirow}
\usepackage{bigdelim}
\usepackage{changepage}
\usepackage{blkarray}
\usepackage{amsmath}
 \usepackage{circuitikz}
\usetikzlibrary{arrows,shapes,chains}
\usetikzlibrary{calc}
\usetikzlibrary{positioning}

\newcommand{\diag}{\mbox{diag}}
\newcommand{\col}{\mbox{col}}

\newtheorem{rem}{Remark}

\newenvironment{proof}{\noindent{\em Proof:\/}}{\hfill $\Box$\par}
\newtheorem{thm}{Theorem}
\newtheorem{lem}{Lemma}

\newtheorem{ass}{Assumption}

\newcommand{\row}{\textnormal{row}}

\newcommand{\myr}{}

\allowdisplaybreaks
\begin{document}\begin{sloppypar}
\begin{frontmatter}
\title{Distributed State Estimation for Jointly Observable Linear Systems over Time-varying Networks \tnoteref{footnoteinfo}}

\tnotetext[t1]{This work was supported by the Natural Sciences and Engineering Research Council (NSERC).}

\author[queen]{Shimin Wang}\ead{shimin.wang@queensu.ca}
\author[queen]{Martin Guay}\ead{martin.guay@queensu.ca}
\address[queen]{Department of Chemical Engineering, Queen's University, Kingston, ON K7L 3N6, Canada.}

\begin{abstract}
This paper deals with a distributed state estimation problem for jointly observable multi-agent systems operated over various time-varying network topologies. The results apply when {\myr the system matrix of the system to be observed contains eigenvalues with positive real parts}. They also can apply to situations where the communication networks are disconnected at every instant. We present sufficient conditions for the existence of distributed observers for general linear systems over periodic communication networks. Using an averaging approach, it is shown that the proposed distributed observer can provide exponentially converging state estimates of the state of the linear system when the network is uniformly connected on average. This average connectedness condition offers a more relaxed assumption that includes periodic switching, Markovian switching and Cox process switching as special cases. All the agents in the network share the estimated state with their neighbours through the network and cooperatively reconstruct the entire state locally. Furthermore, this study presents two exponential stability results for two classes of switched systems, providing valuable tools in related distributed state estimation approaches. A toy example and three practical applications are provided to illustrate the effectiveness of the theoretical results.
\end{abstract}
\end{frontmatter}

\section{Introduction}
The cooperative control of complex systems has received tremendous attention in many areas of investigation, such as distributed sensor fusion in sensor networks \citep*{olfati2005consensus,dougherty2016extremum} and unmanned aerial vehicle formations \citep*{bar1995multitarget}, amongst others.
An important class of cooperative control problems is the leader-follower consensus problem \citep*{su2011cooperative,priscoli2015leader}. 
The objective of this problem is to design a distributed control system that synchronizes the state of the followers with the leader's state \citep*{hong2008distributed}.
One major difficulty with this problem is that the state of the leader is generally not directly available for each follower or some followers.
This problem requires one to design distributed observer techniques that enable state estimation in large-scale dynamical systems operating over a network of spatially distributed sensors \citep*{liu2023distributed}.
These difficulties from cooperative control and practical demands have motivated the study of distributed state estimation problems, as originally described in \cite{olfati2007distributed,hong2008distributed,su2011cooperative} and \cite{su2012cooperative} and fully developed in \cite{wang2017distributed,mitra2018distributed,han2018simple} and \cite{kim2019completely}.

The distributed state estimation problem aims to design distributed observers for each agent (follower system or distributed sensor) to asymptotically reconstruct the state of the (leader or large-scale) system to be observed by using the local measurements and information obtained from its neighbours \citep*{su2011cooperative,wang2017distributed,mitra2018distributed}. %
The existing research addressing the distributed state estimation problem can be classified into the following four categories that reflect the observability assumptions of the system to be observed and the communication network: 1) Disjointly observable systems \citep*{olfati2007distributed}, 2) Semi-jointly observable systems \citep*{su2011cooperative,su2012cooperative}, 3) Locally jointly observable systems \citep*{mitra2018distributed}, and, 4) Jointly observable systems \citep*{wang2017distributed,han2018simple,kim2019completely,yang2022state,yang2023state}.

The disjoint observability assumption allows each agent to independently reconstruct the state of the system to be observed without interaction or communication, implying that the estimation of each agent does not impact the observability of other agents.
A distributed algorithm was constructed in \cite{olfati2007distributed} using a group of sensors
over a static undirected graph under the disjoint observability assumption with cooperative terms to improve the performance of estimation.
For semi-jointly observable systems, some special agents are equipped with sensors that lead to the observability of the full system to be observed. The rest of the agents that do not meet this observability condition require interactions with these special agents. 
%
The semi-jointly observable assumption has been employed for the solution of many leader-following consensus problems and cooperative tracking control problems \citep*{zhang2011optimal,zhang2012adaptive}.
It is worth mentioning that this assumption has also been used in \cite{su2011cooperative} and \cite{su2012cooperative} who proposed a distributed observer scheme that solves a cooperative output regulation problem. 
 The locally jointly observable assumption was initially proposed in \cite{mitra2018distributed} to address situations where each agent and its connected neighbours collectively yield an output matrix that is detectable.
%
The jointly observable assumption provides the mildest possible restriction on the system to be observed and the network. This assumption does not require the full observability of any agent in the network. It allows for the reconstruction of the system to be observed via local interactions among agents and local partial measurements. 

The distributed state estimation problem over jointly observable assumption has been solved under various communication conditions for the linear system case in \cite{wang2017distributed,han2018simple,kim2019completely,wang2022distributed,yang2022state,yang2023state} and \cite{liu2023distributed}, and the nonlinear system case in \cite{wu2021design}.
Specifically, a simple Luenberger-type local observer was presented in \cite{kim2019completely} by using the Kalman observable canonical decomposition, which also promotes some variations of the distributed observer as in \cite{han2018simple,jiao2022mathcal,wang2022distributed} and \cite{wang2022aperiodic} for the static communication network.
{\myr For example, \cite{wang2022distributed} propose a split-spectrum-based distributed estimator for the case when the system dynamics are continuous and the network is stationary.}
{\myr 
Moreover, state estimation for a class of linear time-invariant systems with unknown inputs has been considered in \cite{yang2022state}, which is a meaningful and practical consideration.}

Static communication networks are an ideal condition. In real word applications, the communication networks can change dynamically due to various factors such as time-varying environment, communication distance, failures of links and network congestion.
%
Some efforts have been made to investigate the distributed state estimation problem for continuous-time systems over time-varying graphs \citep*{wang2020distributed,zhang2021distributed,liu2022distributed,yang2023state}.
{\myr The distributed estimation problems for switching networks that are strongly connected at every time with continuous-time and discrete-time systems were tackled in \cite{wang2020distributed} and \cite{wang2022distributed}, respectively.}
In \cite{zhang2021distributed,yang2023state} and \cite{liu2022distributed}, jointly connected switching networks were investigated. This assumption constitutes the weakest requirement on the switching networks. It allows for networks that can be disconnected at every time instant.
{\myr However, we note that the problem considered in \cite{zhang2021distributed,liu2022distributed} and \cite{yang2023state} required the neutral stability of the matrix of the system to be observed, i.e. all the eigenvalues of the system matrix are semi-simple with zero real parts, or all the observer pairs (consisting of the system matrix and the output matrix) are marginally stabilizable, respectively.}

This paper proposes the design of distributed observers for a general linear time-invariant system to estimate the state of the system to be observed subject to jointly observable systems operated over time-varying graphs.
All the agents in the networks cooperate to reconstruct the state of the system to be observed and share the estimated state with their neighbours through local communication. The results allow the observed system to be unstable. The communication network can be disconnected at every instant. This class of networks can include connected static networks or every-time connected switching networks as special cases.
Some sufficient conditions are presented for the existence of distributed observers for general linear systems to be observed over periodic communication networks.
Using an averaging approach, it is shown that the distributed observer can estimate the state of the system to be observed exponentially fast over uniformly connected on average. This approach provides a more relaxed assumption on the network and can include periodic
switching, Markovian switching and, Cox process switching, as special cases.
Two exponential stability results are established for two classes of switched systems.
%
%
%
%
%
%
%
%
%
%
%
%

%
%
%
%
%
%

%
%
%
%
%
%
The rest of the paper is organized as follows. In Section \ref{section2}, we introduce some standard assumptions and lemmas. Section \ref{mainresults} is devoted to the design and analysis of the proposed distributed observers over time-varying networks. This is followed by a simulation study in Section \ref{numerexam}. Four simulation examples are presented which include one toy example and three practical applications. This is followed by brief conclusions in Section~\ref{conlu}.

\textbf{Notation:} Let $\|\cdot\|$ denote both the Euclidean norm of a vector and the Euclidean-induced matrix norm (spectral norm) of a matrix. $\mathds{R}$ is the set of real numbers. $\mathds{N}$ denotes all natural numbers. $I_n$ denotes the $n\times n$ identity matrix. For $A\in \mathds{R}^{m\times n}$, $\textnormal{Ker}(A)=\{x\in \mathds{R}^n|Ax=0\}$ and $\textnormal{Im}(A) =\{y\in \mathds{R}^m | y=Ax \textnormal{~~for~~some~~} x \in \mathds{R}^n\}$ denote the kernel and range of $A$, respectively. For a subspace $\mathcal{V}\subset \mathds{R}^{n}$, the orthogonal complement of $\mathcal{V}$ is denoted as $\mathcal{V}^{\bot}=\{x\in \mathds{R}^{n}|x^Tv=0, \forall v\in \mathcal{V}\}$. $\otimes$ denotes the Kronecker product of matrices.
For $b_i\in \mathds{R}^{n_i \times p}$, $i=1,\dots,m$,
$\col(b_1,\dots,b_m)\triangleq\big[
           \begin{smallmatrix}
             b_1^T & \cdots & b_m^T \\
            \end{smallmatrix}
           \big]^T$. For $a_i\in \mathds{R}^{p \times n_i}$, $i=1,\dots,m$,
$\row(a_1,\dots,a_m)\triangleq\big[
            \begin{smallmatrix}
             a_1 & \cdots & a_m\\
            \end{smallmatrix}
           \big]$.
For $X_1\in \mathds{R}^{n_1\times m_1},\dots,X_k\in \mathds{R}^{n_k\times m_k}$,
\begin{align*}\mbox{diag} (X_1,\dots,X_k)\triangleq\left[
                           \begin{array}{ccc}
                            X_1 &  &  \\
                             & \ddots &  \\
                             &  & X_k\\
                           \end{array}
                          \right].
\end{align*}
For a matrix $P\in \mathds{R}^{n\times n}$, let $\lambda_p$ and $\lambda_P$ denote the minimum and maximum eigenvalues of $P$, respectively.
\section{Problem Formulation and Assumptions}\label{section2}
We consider a linear time-invariant system of the form:
\begin{align}
\dot{x}(t) &= Ax(t),\label{leader}\\
 y_i(t)&=C_ix(t),~~~~i=1,\cdots,N \label{leaderooutput}
 \end{align}
where $x(t)\in \mathds{R}^{n}$ is the vector of state variables, and $y_i\in \mathds{R}^{p_i}$ is the output of system \eqref{leader} detected by agent $i$'s sensor, the known matrices $A$ and $C_i$ are of proper dimension, for $i=1,\cdots,N$.

As in \cite{zhang2021distributed} and \cite{liu2022distributed}, it is assumed that the multi-agent system is composed
of $N$ agents.
The network topology of the multi-agent system is described by a switching digraph ${\mathcal{G}}_{\sigma\left(t\right)}=\left(\bar{\mathcal{V}}, {\mathcal{E}}_{\sigma\left(t\right)}\right)$ where $\sigma\left(t\right)$ is a switching signal with a dwelling time $\tau >0$. For the vertices
 ${\mathcal{V}}=\{1,\cdots,N\}$ , we have that $\left(i,j\right)\in{\mathcal{E}}_{\sigma\left(t\right)}$ if and only if $a_{ji}(t)>0$ at time instant $t$. We denote the set of neighbours of agent $i$ at time $t$ as $\mathcal{\bar{N}}_i (t)$. For further details on graph theory, the reader is referred to \cite{zhang2015constructing}.

The objective of this paper is to design distributed observers over time-varying networks to estimate the state of system \eqref{leader} in the sense that the estimation state $\hat{x}_i(t)$ of each agent's observer converges to the state $x(t)$, i.e.,
$$\lim_{t\rightarrow\infty}\left(\hat{x}_i(t)-{x}(t)\right)=0,~~~~i=1,\cdots,N.$$


Before we proceed, the following assumptions are needed.
\begin{ass}\label{ass0} There exists a {\myr subsequence $\{i_k|k=0,1,\cdots\}$ of $\{i|i=0,1,2,\dots\}$} with $t_{i_{k+1}}-t_{i_k} < \upsilon $ for some positive $\upsilon$ such that the union graph $\bigcup_{j=i_k}^{i_{k+1}-1}{\mathcal{G}}_{\sigma\left(t_j\right)}$ is a strongly connected graph.
\end{ass}
\begin{ass}\label{ass0i}The switching signal $\sigma(t)$ is periodic.
\end{ass}
\begin{rem} Under Assumption \ref{ass0i}, we can assume that the switching signal is as follows:
\begin{align}\label{periodicsw}
\sigma(t)=\begin{cases}
          1, &\text{If $sT \leq t< (s+\omega_1)T$};\\
          2, &\text{If $(s+\omega_1)T \leq t< (s+\sum_{\rho=1}^2\omega_\rho)T$};\\
         \;\vdots & \;\;\;\;\;\;\;\;\;\;\;\;\;\;\;\;\;\;\;\;\;\;\;\;\;\;\;\;{\centering \vdots} \\
          p, &\text{If $(s+\sum_{\rho=1}^{p-1}\omega_\rho)T \leq t< (s+1)T$};
        \end{cases}
\end{align}
where $T$ is a positive constant, $s=0,1,2,\cdots,$ and $\omega_\rho$, $\rho=1,\cdots,p$ are positve constants satisfying $\sum_{\rho=1}^{p}\omega_p=1$. Under Assumptions \ref{ass0} and \ref{ass0i}, for any $k=0,1,2,\cdots$, the union graph $\bigcup_{j=i_k}^{i_{k+1}-1}{\mathcal{G}}_{\sigma\left(t_j\right)}=\bigcup_{\rho=1}^{p}{\mathcal{G}}_{\rho}\equiv{\mathcal{G}}$. Let $\mathcal{L}$ denote the Laplacian matrix associated with graph ${\mathcal{G}}$.
\end{rem}
A more relaxed assumption is the so-called uniformly
connected on average networks property \citep*{wang2010input,kim2013consensus,stilwell2006sufficient}. This assumption is shared by many time-varying networks with switching topologies such as periodic switching, Markovian switching and Cox process switching networks.
\begin{ass} \label{assuniform} A graph $\mathcal{G}_{\sigma(t)}$ with its corresponding Laplacian matrix $\mathcal{L}_{\sigma(t)}$, is said to be connected on average with uniform convergence to the average (for short, uniformly connected on average) if the average Laplacian
 $$\mathcal{L}=\lim_{T\rightarrow\infty}\frac{1}{T}\int_{t}^{t+T}\mathcal{L}_{\sigma(\tau)}d\tau,\;t\geq 0.$$
induces a graph $\mathcal{G}$ that is connected.
\end{ass}
\begin{rem}\label{remark2} The convergence of the limit is uniform with respect to $t$ in the sense that there is a continuous and strictly decreasing function $\delta:[0,\infty)$, called a convergence function, such that $\lim_{T\rightarrow\infty}\delta(T)=0$ and
$$\Big\|\frac{1}{T}\int_{t}^{t+T}\mathcal{L}_{\sigma(\tau)}d\tau-\mathcal{L}\Big\|\leq\delta(T),\;t\geq 0. $$
Motivated by \cite{wang2010input,kim2013consensus} and \cite{stilwell2006sufficient}, we consider an averaging approach to describe the fast-switching network. We consider the time-varying Laplacian matrix of the form $\mathcal{L}_{\sigma(t/\epsilon)}$, where the parameter $\epsilon> 0 $ determines the switching speed of the network topology.
\end{rem}
\begin{ass}\label{ass1} The system defined in \eqref{leader} is jointly observable
\end{ass}

\begin{rem}For $i\in \mathcal{V}$, we assume that the observability index of {\myr $(C_i, A)$} is $v_i$, such that $\textnormal{rank}(\mathcal{O}_i)=v_i$, 
where $\mathcal{O}_i \in \mathds{R}^{np_i\times n}$ is the observability matrix given by: $\mathcal{O}_i=\textnormal{\col}\left(C_i,C_iA,\cdots,C_iA^{n-1}\right)$. 
For $i\in \mathcal{V}$, the observable subspace and unobservable subspace of {\myr $(C_i, A)$} are defined as $\textnormal{Im}(\mathcal{O}_i^T)\subset \mathds{R}^{n}$ and $\textnormal{Ker}(\mathcal{O}_i)\subset \mathds{R}^{n}$, respectively, and satisfy $\textnormal{Ker}(\mathcal{O}_i)^{\bot}=\textnormal{Im}(\mathcal{O}_i^T)$.

For $i\in \mathcal{V}$, let $V_i=\row(V_{ui},V_{oi})\in\mathds{R}^{n\times n}$ be an orthogonal matrix such that $V_iV_i^T=I_n$.
Let $V_{ui}\in\mathds{R}^{n\times (n-v_i)}$ be a matrix such that all columns of $V_{ui}$ are from an orthogonal basis of the $\textnormal{Ker}(\mathcal{O}_i)$ satisfying $\textnormal{Im}(V_{ui})=\textnormal{Ker}(\mathcal{O}_i)$.
Let $V_{oi}\in\mathds{R}^{n\times v_i}$ be a matrix such that all columns of $V_{oi}$ are from an orthogonal basis of the $\textnormal{Im}(\mathcal{O}_i^T)$ satisfying $\textnormal{Im}(V_{oi})=\textnormal{Im}(\mathcal{O}_i^T)$.
%

%
For $i\in \mathcal{V}$, we use the Kalman observability decomposition to express the matrices $A$ and $C_i$ of the system in \eqref{leader} as follows:
\begin{subequations}\label{decom}\begin{align}
V_i^{T}AV_i
=&\left[
      \begin{array}{cc}
       A_{ui} & A_{ri} \\
       0_{v_i\times (n-v_i)} & A_{oi} \\
      \end{array}
     \right],\\
C_iV_i=&\left[
               \begin{array}{cc}
                0_{p_i\times (n-v_i)} & C_{oi} \\
               \end{array}
              \right],
\end{align}\end{subequations}
where the pair {\myr $(C_{oi}, A_{oi})$} is observable, $A_{oi}\in \mathds{R}^{v_i\times v_i}$, $ A_{ri}\in \mathds{R}^{(n-v_i)\times v_i}$, $A_{ui}\in \mathds{R}^{(n-v_i)\times (n-v_i)}$ and $ C_{oi} \in \mathds{R}^{p_i\times v_i}$ admit the following matrix decomposition: $A_{ui}=V_{ui}^TAV_{ui}$, $A_{ri}=V_{ui}^TAV_{oi}$, $A_{oi}=V_{oi}^TAV_{oi}$ and $C_{oi}=C_iV_{oi}$. \end{rem}

 Let $C_{o}=\textnormal{diag}({C_{o1}},\cdots,{C_{oN}})$, $A_{r}=\textnormal{diag}({A_{r1}},\cdots,{A_{rN}})$, $V_{o}=\textnormal{diag}(V_{o1},\cdots,V_{oN})$, $V_{u}=\textnormal{diag}(V_{u1},\cdots,V_{uN})$, $A_{o}=\textnormal{diag}({A_{o1}},\cdots,{A_{oN}})$ and $A_{u}=\textnormal{diag}({A_{u1}},\cdots,{A_{uN}})$.

Next, we state some useful lemmas proposed in \cite{sun2006switched,kim2019completely} and \cite{qu2009cooperative}
\begin{lem} \citep*{sun2006switched}\label{lemmasun} Consider the linear switched system given by:
\begin{align}\label{lsitch}
\dot{y}=A_{\sigma(t)}y
\end{align}
where $y\in \mathds{R}^{n}$ is the vector of state variables, $\sigma:[0, \infty)\mapsto \mathcal{P}=\{1,2,\cdots,p\}$ is the switching signal satisfying Assumption \ref{ass0i}, and $A_\rho\in \mathds{R}^{n\times n}$, $\rho=1\cdots,p$. If the matrix $\sum_{\rho=1}^{p}\omega_pA_p$ is Hurwitz with $\omega_\rho$, $\rho=1,\cdots,p$, obtained from \eqref{periodicsw}, then there exists a positive constant $\bar{T}_0$ such that for $0<T<\bar{T}_0$, the origin of system \eqref{lsitch} is exponentially stable.
\end{lem} 
\begin{lem}\label{meitopol}\citep*{zhang2015constructing} Suppose that the communication network $\mathcal{G}=(\mathcal{V},\mathcal{E})$ is strongly connected. Let $\theta=\col(\theta_1,\cdots,\theta_N)$ be the left eigenvector of the Laplacian matrix $\mathcal{L}$ associated with the eigenvalue $0$, i.e., $\mathcal{L}^T\theta=0$. Then,
 $\Theta =\textnormal{\diag}(\theta_1,\cdots,\theta_N)>0$ and
$\hat{\mathcal{L}}=\Theta \mathcal{L}+\mathcal{L}^T\Theta\geq0$.
\end{lem}
\begin{lem}\label{kimPD}\citep*{kim2019completely} Suppose that the communication network $\mathcal{G}=(\mathcal{V},\mathcal{E})$ is strongly connected. Then, the following statements are equivalent:
\begin{enumerate}
 \item System \eqref{leader} is jointly observable;
 \item The matrix $V_{u}^T\big(\hat{\mathcal{L}}\otimes I_n\big)V_{u}$ is positive definite;
 \item The matrix $V_{u}^T\left({\mathcal{L}}\otimes I_n\right)V_{u}$ is nonsingular.
\end{enumerate}
\end{lem}
 Under Assumption \ref{ass0} and \ref{ass1}, the matrix $V_{u}^T(\mathcal{\hat{L}}\otimes I_n)V_{u}$ is positive definite matrix from Lemma \ref{kimPD}. Let $\lambda_{l}$ and $\lambda_{L}$ denote the minimum and maximum eigenvalues of $V_{u}^T(\mathcal{\hat{L}}\otimes I_n)V_{u}$, respectively. Let $V_{u}=\textnormal{diag}(V_{u1},\cdots,V_{uN})$ and $A_{u}=\textnormal{diag}({A_{u1}},\cdots,{A_{uN}})$. Then, we establish the following lemmas.
 \begin{lem}\label{lemmaexp} Under Assumption \ref{ass1}, consider the following linear time-invariant system
 \begin{align}\label{compactformii}
\dot{z} =&\big[A_{u}-\gamma V_{u}^T\left(\mathcal{L}\otimes I_n\right)V_{u}\big]z,
\end{align}
where $z\in \mathds{R}^{\sum_{i=1}^{N}(n-\nu_i)}$. Suppose that the communication network $\mathcal{G}$ is strongly connected. Then, the system \eqref{compactformii} satisfies the following properties:
\begin{enumerate}
\item System \eqref{compactformii} is globally asymptotically stable for a sufficiently large value of $\gamma$.
\item The matrix $A_{u}-\gamma V_{u}^T\left(\mathcal{L}\otimes I_n\right)V_{u}$ is Hurwitz for a sufficiently large value of $\gamma$.
\end{enumerate}
 \end{lem}
 \begin{proof} Define the following Lyapunov function candidate for system \eqref{compactformii}
\begin{align}\label{lyapu}
V(z)=\sum\nolimits_{i=1}^{N}\theta_iz_i^T z_i.
\end{align}
Then,
\begin{equation}\label{temp}
\theta_m\|z\|^2\leq V(z)\leq\theta_M\|z\|^2,
\end{equation}
where $\theta_m=\min\{\theta_1,\cdots,\theta_N\}$ and $\theta_M=\max\{\theta_1,\cdots,\theta_N\}$.
The time derivative of $V(z)$ along (\ref{compactformii}) can be evaluated as
\begin{align}
\dot{V} =&2\sum\nolimits_{i=1}^{N}\theta_i z_i^T A_{ui}z_i -\gamma z^T\big[V_{u}^T(\mathcal{\hat{L}}\otimes I_n)V_{u}\big]z\nonumber\\
\leq&2\theta_M\|A\|\|z\|^2- \gamma \lambda_l\|z\|^2.\nonumber
\end{align}
It is easily to verified that $\|A_{ui}\|\leq\|A\|$, for $i=1,\cdots,N$. From \eqref{compactformii} and \eqref{temp}, we have
\begin{align}
\dot{V} \leq&\frac{2\theta_M^2\|A\|-\gamma\lambda_l\theta_m}{\theta_M\theta_m}V.\nonumber
\end{align}
Hence, for any $\gamma>\frac{2\theta_M^2\|A\|}{\lambda_l\theta_m}$, system \eqref{compactformii} is globally asymptotically stable (exponentially stable). From the Theorem 4.5 in \cite{khalil2002nonlinear}, the matrix $A_{u}-\gamma V_{u}^T\left(\mathcal{L}\otimes I_n\right)V_{u}$ is Hurwitz for sufficiently large enough $\gamma$. 
 \end{proof}
  \begin{lem}\label{avlemmaexp0}Under Assumptions \ref{ass0}, \ref{ass0i} and \ref{ass1}, consider the following linear system
 \begin{align}\label{avcompactformiipe}
\dot{z} =&\big[A_{u}-\gamma V_{u}^T\left(\mathcal{L}_{\sigma(t)}\otimes I_n\right)V_{u}\big]z,
\end{align}
where $z=\textnormal{\col}(z_1,\cdots,z_N)$ with $z_i\in \mathds{R}^{n-\nu_i}$. Then there exist some constants $\bar{\gamma}_0$ and $\bar{T}_0$, such that, for all $\gamma\geq \bar{\gamma}_0$ and $0<T<\bar{T}_0$, the system \eqref{avcompactformiipe} is exponentially stable at the origin.
 \end{lem}
 \begin{proof}
Under Assumptions \ref{ass0} and \ref{ass0i}, for any $k=0,1,2,\cdots$, the union graph $\bigcup_{j=i_k}^{i_{k+1}-1}{\mathcal{G}}_{\sigma\left(t_j\right)}=\bigcup_{\rho=1}^{p}{\mathcal{G}}_{\rho}={\mathcal{G}}$
is a strongly connected directed graph. Clearly, $\mathcal{L}=\sum_{\rho=1}^{p}\omega_\rho{\mathcal{L}}_{\rho}$. Then, form $\sum_{\rho=1}^{p}\omega_\rho=1$, we have
$$\sum_{\rho=1}^{p}\omega_\rho\big[A_{u}-\gamma V_{u}^T\left(\mathcal{L}_{\rho}\otimes I_n\right)V_{u}\big]=A_{u}-\gamma V_{u}^T\left(\mathcal{L}\otimes I_n\right)V_{u}.$$
Thus, from Lemma \ref{lemmaexp}, there exists $\bar{\gamma}_0$ such that $A_{u}-\gamma V_{u}^T\left(\mathcal{L}\otimes I_n\right)V_{u}$ is is Hurwitz for any $\gamma\geq\bar{\gamma}_0$. Then, from Lemma \ref{lemmasun}, there exists a positive constant $\bar{T}_0$ such that for any $0<T<\bar{T}_0$, system \eqref{avcompactformiipe} is exponentially stable at the origin.
\end{proof}
 \begin{lem}\label{avlemmaexp} Under Assumptions \ref{ass0}, \ref{assuniform} and \ref{ass1}, consider the following linear system
 \begin{align}\label{avcompactformii}
\dot{z} =&\big[A_{u}-\gamma V_{u}^T\left(\mathcal{L}_{\sigma(t/\epsilon)}\otimes I_n\right)V_{u}\big]z,
\end{align}
where $z=\textnormal{\col}(z_1,\cdots,z_N)$ with $z_i\in \mathds{R}^{n-\nu_i}$. Suppose that the communication network $\mathcal{G}$ is strongly connected. Then, there exists a positive scalar $\varepsilon^*$ such that, for any $0<\epsilon<\varepsilon^*$ and initial condition, the system \eqref{avcompactformii} is exponentially stable at the origin.
 \end{lem}
\begin{proof} The proof proceeds following the same steps as in \cite{khalil2002nonlinear}and \cite{kim2013consensus}.
We define $Q_{av}= A_{u}-\gamma V_{u}^T\left(\mathcal{L}\otimes I_n\right)V_{u}$ and ${\myr Q(t/\epsilon)= \gamma V_{u}^T\big[(\mathcal{L}-\mathcal{L}_{\sigma(t/\epsilon)})\otimes I_n\big]V_{u}}$. Then, we have
 \begin{align}
\dot{z} =&Q_{av}z+ Q(t/\epsilon)z.\nonumber
\end{align}
{\myr Under Assumption \ref{assuniform}, and from Remark \ref{remark2}}, $Q(t)$ satisfies
$$\Big\|\frac{1}{T}\int_{t}^{t+T}Q(\tau)d\tau\Big\|\leq k_1\delta(T),$$
where $k_1=\gamma \max\limits_{i\in \mathcal{V}}\|V_{ui}\|^2$. 
%
%
%
Then, we define the following matrix $$w(t,\epsilon)=\int_{0}^{t}Q(\tau)e^{-\epsilon(t-\tau)}d\tau,$$
for some positive constant $\epsilon$. At $\epsilon=0$, the function $w(t,0)$ satisfies
 \begin{align}\label{av2}
 \|w(t_1, 0)-w(t_2,0)\|
 \leq&k_1|t_2-t_1|\delta(t_2-t_1),
 \end{align}
 where $k_1=\gamma\max\limits_{i\in \mathcal{V}}\|V_{ui}\|^2$. Integrating $w(t, \epsilon)$ by parts, we have
 \begin{align}\label{av3}
 w(t,& \epsilon)= w(t, 0)-\epsilon\int_{0}^{t}e^{-\epsilon(t-\tau)}w(\tau, 0)d\tau\nonumber\\
 =&{\myr w(t, 0)-\epsilon\int_{0}^{t}e^{-\epsilon(t-\tau)}w(t, 0)d\tau}\nonumber\\
 &{\myr-\epsilon\int_{0}^{t}e^{-\epsilon(t-\tau)}\big[w(\tau, 0)-w(t, 0)\big]d\tau}\nonumber\\
  =&{\myr w(t, 0)-w(t, 0)\left[e^{-\epsilon(t-\tau)}\big|_{0}^{t}\right]}\nonumber\\
 &{\myr-\epsilon\int_{0}^{t}e^{-\epsilon(t-\tau)}\big[w(\tau, 0)-w(t, 0)\big]d\tau}\nonumber\\
 =&{\myr w(t, 0)-w(t, 0)\left[1-e^{-\epsilon t}\right]}\nonumber\\
 &{\myr-\epsilon\int_{0}^{t}e^{-\epsilon(t-\tau)}\big[w(\tau, 0)-w(t, 0)\big]d\tau}\nonumber\\
 = & e^{-\epsilon t}w(t, 0)-\epsilon\int_{0}^{t}e^{-\epsilon(t-\tau)}\big[w(\tau, 0)-w(t, 0)\big]d\tau.
  \end{align}
 Using the last expression, it follows that $w(t,\epsilon)$ fulfills the inequality:
 {\myr\begin{align}\label{av4}
 \| w(t, &\epsilon)\|\leq te^{-\epsilon t}k_1\delta(t)+k_1\epsilon\int_{0}^{t}e^{-(t-\tau)\epsilon}(t-\tau)\delta(t-\tau)d\tau.
  \end{align}}
It follows {\myr from \cite{khalil2002nonlinear}} that there exits a class $\mathcal{K}$ function $\kappa$ such that $ \epsilon\| w(t, \epsilon)\|\leq \kappa(\epsilon)$ for all $s\leq0$ and $0\leq \epsilon\leq 1$. Therefore, $\epsilon\| w(t, \epsilon)\|=O(\kappa(\epsilon))$\footnote{$f_{1}(\epsilon)=O(f(\epsilon))$ if there are positive constants $a$ and $b$ such that $\|f_{1}(\epsilon)\|\leq a\|f(\epsilon)\|$ for $\|\epsilon\|< b$.}.
Let us define a new variable {\myr $\tilde{z}$ satisfying} $$\myr z=[I+\epsilon w(t/\epsilon, \epsilon)]\tilde{z}.$$ 
As $\epsilon\| w(t/\epsilon, \epsilon)\|=O(\kappa(\epsilon))$ for all $t\geq 0$, there exists a constant $\epsilon_1>0$ such that $I + \epsilon w(t/\epsilon, \epsilon) $ is nonsingular for $0\leq\epsilon\leq \epsilon_1$. Denote
 $$\big[I + \epsilon w(t/\epsilon, \epsilon)\big]^{-1}\triangleq I+ F(t/\epsilon,\epsilon).$$
It can be seen that $F(t/\epsilon,\epsilon)=O( \kappa(\epsilon))$ for all $0\leq t$. The coordinate change yields the following dynamics
  \begin{align}\label{av5}
  \dot{z}=&\epsilon \dot{w}(t/\epsilon, \epsilon){\myr \tilde{z}}+[I+\epsilon w(t/\epsilon, \epsilon)]{\myr \dot{\tilde{z}}}\nonumber\\
  =&\epsilon\frac{d\left(\int_{0}^{t/\epsilon}Q(\tau)e^{-\epsilon(t/\epsilon-\tau)}d\tau\right)}{dt} \tilde{z} +[I+\epsilon w(t/\epsilon, \epsilon)]\dot{\tilde{z}}\nonumber\\
    =&\epsilon\left(Q(t/\epsilon)/\epsilon-\int_{0}^{t/\epsilon}Q(\tau)e^{-\epsilon(t/\epsilon-\tau)}d\tau\right) \tilde{z} +[I+\epsilon w(t/\epsilon, \epsilon)]\dot{\tilde{z}}\nonumber\\
  =& \big[ Q(t/\epsilon)-\epsilon w(t/\epsilon, \epsilon)\big] {\myr \tilde{z}}+[I+\epsilon w(t/\epsilon, \epsilon)]{\myr \dot{\tilde{z}}}.
  \end{align}
 Then, we have
{ \myr\begin{align}\label{av5} 
 \dot{\tilde{z}}=& [I+\epsilon w(t/\epsilon, \epsilon)]^{-1}\dot{z}\nonumber\\
 &-[I+\epsilon w(t/\epsilon, \epsilon)]^{-1}\big[ Q(t/\epsilon) -\epsilon w(t/\epsilon, \epsilon) \big]\tilde{z} \nonumber\\
 =& [I+\epsilon w(t/\epsilon, \epsilon)]^{-1}\left[Q_{av}+ Q(t/\epsilon)\right][I+\epsilon w(t/\epsilon, \epsilon)]\tilde{z}\nonumber\\
 &-[I+\epsilon w(t/\epsilon, \epsilon)]^{-1}\big[ Q(t/\epsilon) -\epsilon w(t/\epsilon, \epsilon) \big]\tilde{z} \nonumber\\
 =&[I+ F(t/\epsilon,\epsilon)]\left[Q_{av}+ Q(t/\epsilon)\right][I+\epsilon w(t/\epsilon, \epsilon)]\tilde{z}\nonumber\\
 &-[I+ F(t/\epsilon,\epsilon)]\big[ Q(t/\epsilon) -\epsilon w(t/\epsilon, \epsilon) \big]\tilde{z} \nonumber\\
 =&[I+ F(t/\epsilon,\epsilon)]Q_{av}[I+\epsilon w(t/\epsilon, \epsilon)]\tilde{z}\nonumber\\
  &+[I+ F(t/\epsilon,\epsilon)] [Q(t/\epsilon)\tilde{z}+\epsilon Q(t/\epsilon) w(t/\epsilon, \epsilon)\tilde{z}]\nonumber\\
 &-[I+ F(t/\epsilon,\epsilon)]\big[ Q(t/\epsilon)\tilde{z} -\epsilon w(t/\epsilon, \epsilon)\tilde{z} \big] \nonumber\\ 
  =&Q_{av}y+ F(t/\epsilon,\epsilon)Q_{av}\tilde{z}\nonumber\\
  &+\epsilon [I+ F(t/\epsilon,\epsilon)]Q_{av}w(t/\epsilon, \epsilon)\tilde{z}\nonumber\\
  &+\epsilon[I+ F(t/\epsilon,\epsilon)] Q(t/\epsilon) w(t/\epsilon, \epsilon)\tilde{z}\nonumber\\
 &+\epsilon[I+ F(t/\epsilon,\epsilon)] w(t/\epsilon, \epsilon)\tilde{z} \nonumber\\ 
 =& \big[Q_{av} +N(t/\epsilon,\epsilon)\big]\tilde{z},
 \end{align}}
 where
 { \myr \begin{align}
 N(t/\epsilon,\epsilon)=& F(t/\epsilon,\epsilon)Q_{av}+\epsilon [I+ F(t/\epsilon,\epsilon)]Q_{av}w(t/\epsilon, \epsilon)\nonumber\\
 &+\epsilon[I+ F(t/\epsilon,\epsilon)] [I+Q(t/\epsilon)] w(t/\epsilon, \epsilon) .\nonumber
 \end{align}}
 By Lemma \ref{lemmaexp}, it can be concluded that $Q_{av}$ is Hurwitz for a sufficiently large enough value of $\gamma$. As a result, there exists a positive definite symmetric matrix $P_{av}$ such that $$P_{av}Q_{av}+Q_{av}^TP_{av}\leq-I.$$
 We pose the following Lyapunov function candidate for system \eqref{av5}:
\begin{align}\label{lyapu} \myr 
U(\tilde{z})=\tilde{z}^TP_{av}\tilde{z}.
\end{align}
The time derivative of {\myr $U(\tilde{z})$} along the trajectories of \eqref{compactformii} can be evaluated as
{\myr\begin{align}\label{lyapu1}
\dot{U}(\tilde{z})=&2\tilde{z}^TP_{av}\big[Q_{av} +N(t/\epsilon,\epsilon)\big]\tilde{z}\nonumber\\
\leq&-\tilde{z}^T\tilde{z}+2\tilde{z}^TP_{av}N(t/\epsilon,\epsilon)\tilde{z}\nonumber\\
\leq&-\|\tilde{z}\|^2+2\|P_{av}\|\|N(t/\epsilon,\epsilon)\|\|\tilde{z}\|^2.
\end{align}}
It is noted that $N(t/\epsilon,\epsilon)=O(\kappa(\epsilon))$. Then, we have $\|N(t/\epsilon,\epsilon)\|\leq N^*\ \kappa(\epsilon)$ for some positive constant $N^*$. As a result, it follows from \eqref{lyapu1} that we have
\begin{align}
{\myr\dot{U}(\tilde{z})\leq-\Big(1-2N^*\ \kappa(\epsilon)\|P_{av}\|\Big)\|\tilde{z}\|^2.}\nonumber
\end{align}
Hence, for any $\epsilon\in (0, \epsilon^*)$ with $\epsilon^*=\min\Big\{\kappa^{-1}\big(\frac{1}{2\|P_{av}\|N^*}\big),\epsilon_1\Big\}$, we have $\dot{U}(\tilde{z})<0$. Therefore, the system \eqref{av5} is exponentially stable at the origin for sufficiently small $\epsilon\in (0, \epsilon^*)$, {\myr together with $\left(Q_{av}z+ Q(t/\epsilon)z\right)\big|_{z=0}\equiv0$ for all $(t,\epsilon)\in [0,\infty)\times [0, \epsilon^*]$} further implies that \eqref{avcompactformii} is exponentially stable at the origin for sufficiently small $\epsilon\in (0, \epsilon^*)$ from {\myr Theorem 10.5} in \cite{khalil2002nonlinear}.
\end{proof}

\section{Main Results}\label{mainresults}
\subsection{Distributed Observer over Time-Varying Graphs}\label{mainresults1}
We now introduce the following linear dynamic observer:
\begin{align}
\dot{\hat{x}}_i=&A\hat{x}_i+ L_i(C_i\hat{x}_i-y_i)+\gamma M_i\sum\limits_{j\in \mathcal{N}_i(t)}(\hat{x}_j-\hat{x}_i),\label{compensator}
\end{align}
 where, for $i=1\cdots,N$, $\hat{x}_i\in \mathds{R}^{n}$ is the estimate of $x$,
 \begin{align}
 L_i=V_i\left[
     \begin{array}{c}
      0 \\
       L_{oi}\\
     \end{array}
    \right],~~M_i=V_i\left[
              \begin{array}{cc}
               I_{n-v_i}& 0 \\
               0 &0\\
              \end{array}
             \right]V_i^T,\nonumber
\end{align}
 $\gamma$ is a sufficiently large positive constant to be determined and $L_{oi}\in \mathds{R}^{v_i\times p_i}$ is chosen such that $(A_{oi}+L_{oi}C_{oi})$ is Hurwitz.

For $i=1,\cdots,N$, let $\tilde{x}_i=\hat{x}_i-x$ be the estimation error of the $i$th observer. Then, we have
\begin{align}\label{decom1}
\dot{\tilde{x}}_i=&A\tilde{x}_i+ L_iC_i\tilde{x}_i+\gamma M_i\sum\nolimits_{j\in\mathcal{N}_i(t)}{(\tilde{x}_j-\tilde{x}_i)}\nonumber\\
=&(A+ L_iC_i)\tilde{x}_i-\gamma M_i\sum\nolimits_{j=1}^{N} l_{ij}(t){\tilde{x}_j},
\end{align}
where $l_{ij}(t)$ is the $(i,j)$-th entry of the Laplacian matrix $\mathcal{L}_{\sigma(t)}$ at time moment $t$.
Let $\tilde{x}_{oi}=V_{oi}^T\tilde{x}_i$ and $\tilde{x}_{ui}=V_{ui}^T\tilde{x}_i$, for $i=1,\cdots,N$. Then, we have the following system from \eqref{decom} and \eqref{decom1},
\begin{subequations}\label{decom2}\begin{align}
\dot{\tilde{x}}_{ui}=&A_{ui}\tilde{x}_{ui}+A_{ri}\tilde{x}_{oi}\nonumber\\
&-\gamma V_{ui}^T\sum\nolimits_{j=1}^{N} l_{ij}(t)\big[{V_{uj}\tilde{x}_{uj}+V_{oj}\tilde{x}_{oj}}\big],\label{decom2b}\\
\dot{\tilde{x}}_{oi}=&(A_{oi}+L_{oi}C_{oi})\tilde{x}_{oi}.\label{decom2a}
\end{align}
\end{subequations}
Let $\tilde{x}_{u}=\col(\tilde{x}_{u1},\cdots,\tilde{x}_{uN})$, $\tilde{x}_{o}=\col(\tilde{x}_{o1},\cdots,\tilde{x}_{oN})$, $V_{o}=\textnormal{diag}(V_{o1},\cdots,V_{oN})$, and $A_{r}=\textnormal{diag}({A_{r1}},\cdots,{A_{rN}})$. Then, the system \eqref{decom2b} can be put into the following compact form,
\begin{align}\label{compactform}
\dot{\tilde{x}}_{u}=&A_{u}\tilde{x}_{u}+A_r\tilde{x}_{o}-\gamma V_{u}^T\left(\mathcal{L}_{\sigma(t)}\otimes I_n\right)\big[V_{u}\tilde{x}_{u}+V_{o}\tilde{x}_{o}\big]\nonumber\\
=&M(t)\tilde{x}_{u}+N(t)\tilde{x}_{o}
\end{align}
where $M(t)=A_{u}-\gamma V_{u}^T\left(\mathcal{L}_{\sigma(t)}\otimes I_n\right)V_{u}$ and $N(t)=A_r-\gamma V_{u}^T\left(\mathcal{L}_{\sigma(t)}\otimes I_n\right)V_{o}$.

Then, we have the following results.
\begin{lem}\label{discrettimeexo} Under Assumption \ref{ass1}, consider systems \eqref{leader} and linear switched system \eqref{compensator}. Then,
$$ \lim_{t\rightarrow\infty}(\hat{x}_i(t)-x(t))=0,\;i=1,\cdots,N,$$ for any $x(0)\in \mathds{R}^{n}$ and $\hat{x}_i(0)\in \mathds{R}^{n}$, provided that the following system \begin{align}\label{compactform2}
\dot{\tilde{x}}_{u}=&M(t)\tilde{x}_{u},
\end{align}
is exponentially stable.
\end{lem}
\begin{proof} As $\tilde{x}_{o}(t)=0$ for all $t\geq 0$, the system \eqref{compactform} reduces to the system \eqref{compactform2}.
Since the system \eqref{compactform2} is exponentially stable, it follows that, for any positive definite matrix $Q(t)$ satisfying $\|Q(t)\|\geq c_3$, for some positive constant $c_3$, there exists a positive definite matrix $P(t)$ satisfying $c_1\leq \|P(t)\|\leq c_2$ for some positive constants $c_1$ and $c_2$ which solves the matrix differential equation:
\begin{align}\label{lypP}
\dot{P}(t)=-P(t)M(t)-M^T(t)P(t)-Q(t).
\end{align}
We pose the following Lyapunov function candidate $V(\tilde{x}_{u})=\tilde{x}_{u}^TP(t)\tilde{x}_{u}$. Its rate of change along the trajectories of system \eqref{compactform} is such that
\begin{align}\label{VtimeP}
\dot{V}=&\tilde{x}_{u}^T\big[\dot{P}(t)+P(t)M(t)+M^T(t)P(t)\tilde{x}_{u}+2P(t)N(t)\tilde{x}_{o}\big]\nonumber\\
=&-\tilde{x}_{u}^TQ(t)\tilde{x}_{u}+2\tilde{x}_{u}^TP(t)N(t)\tilde{x}_{o}\nonumber\\
\leq&-c_3\|\tilde{x}_{u}\|^2+2 c_2\|N(t)\|\|\tilde{x}_{u}\|\|\tilde{x}_{o}\|\nonumber\\
\leq&-c_3\|\tilde{x}_{u}\|^2+\bigg[\frac{c_3}{4}\|\tilde{x}_{u}\|^2+\frac{4c_2^2}{c_3}\|N(t)\|^2\|\tilde{x}_{o}\|^2\bigg]\nonumber\\
\leq&-\frac{3c_3}{4}\|\tilde{x}_{u}\|^2+\frac{4c_2^2}{c_3}\|N(t)\|^2\|\tilde{x}_{o}\|^2\nonumber\\
\leq&-\frac{3c_3}{4c_2}V+\frac{4c_2^2}{c_3}\|N(t)\|^2\|\tilde{x}_{o}\|^2.
\end{align}
From Assumption \ref{ass1}, $L_{oi}$ was chosen such that $A_{oi}+L_{oi} C_{oi}$ is Hurwitz, for $i=1,\cdots,N$, and, therefore, we conclude that $\lim_{t\rightarrow \infty}\tilde{x}_{oi}(t)=0$. As a result, we can also conclude that $\lim_{k\rightarrow \infty}\|\tilde{x}_{o}\|^2=0$. We note that the matrix ${\mathcal{L}}_{\sigma(t)}$ is a piecewise constant matrix where $\sigma(t)$ has the range $\mathcal{P}=\{1,\cdots,p\}$. Thus, $N(t)$ is bounded over $[0, +\infty)$ and continuous on each time interval $[t_{s}, t_{s+1})$, $s=0,1,\cdots$. It follows that, for all $t\geq0$, $\lim_{t\rightarrow\infty}\frac{4c_2^2}{c_3}\|N(t)\|^2\|\tilde{x}_{o}(t)\|^2=0$ exponentially. Hence, system \eqref{lypP} can be viewed as an input to state stable system with $\frac{4c_2^2}{c_3}\|N(t)\|^2\|\tilde{x}_{o}(t)\|^2$ as the input. Since this input tends to zeros, it follows that we can conclude that $\lim_{t\rightarrow\infty}V(t)=0$ exponentially, which further implies $\lim_{t\rightarrow\infty}\tilde{x}_{o}(t)=0$. Hence, $\lim\limits_{t\rightarrow\infty}\tilde{x}_{ui}(t)=0$ and $\lim\limits_{t\rightarrow\infty}\tilde{x}_{oi}(t)=0$. And, since, $\tilde{x}_{oi}(t)=V_{oi}^T\tilde{x}_i(t)$ and $\tilde{x}_{ui}(t)=V_{ui}^T\tilde{x}_i(t)$, we conclude that $\lim\limits_{t\rightarrow\infty}\tilde{x}_i(t)=0$, for $i=1,\cdots,N$.
\end{proof}

\begin{thm}\label{thm1xe} Consider the systems \eqref{leader} and the linear switched system \eqref{compensator}. Under Assumptions \ref{ass0}, \ref{ass0i} and \ref{ass1}, there exists positive $\bar{\gamma}_0>0$ and $\bar{T}_0$ such that for any $\gamma\geq \bar{\gamma}_0$ and $0<T<\bar{T}_0$,
$ \lim_{t\rightarrow\infty}(\hat{x}_i(t)-x(t))=0$, exponentially for any $x(0)\in \mathds{R}^{n}$ and $\hat{x}_i(0)\in \mathds{R}^{n}$, $i=1,\cdots,N$.
\end{thm}
\begin{proof} In order to analyze the system \eqref{compactform}, we first assume that $\tilde{x}_{o}(t)=0$ for all $t\geq 0$, the system \eqref{compactform} will reduce to the system
\begin{align}\label{tldeAu}
\dot{\tilde{x}}_{u}=\left(A_{u}-\gamma V_{u}^T\left(\mathcal{L}_{\sigma(t)}\otimes I_n\right)V_{u}\right)\tilde{x}_{u}
\end{align}
System \eqref{tldeAu} is in the form of \eqref{avcompactformiipe}. Under Assumptions \ref{ass0}, \ref{ass0i} and \ref{ass1},
from Lemma \ref{avlemmaexp0}, there exists positive $\bar{\gamma}_0>0$ and $\bar{T}_0$ such that for any $\gamma\geq \bar{\gamma}_0$ and $0<T<\bar{T}_0$, \eqref{tldeAu} is exponentially stable at origin.
Therefore,
by using Lemma \ref{discrettimeexo}, we have $ \lim_{t\rightarrow\infty}(\hat{x}_i(t)-x(t))=0$, exponentially for any $x(0)\in \mathds{R}^{n}$ and $\hat{x}_i(0)\in \mathds{R}^{n}$, $i=1,\cdots,N$.
\end{proof}

\begin{thm} \label{theo2} Consider the systems \eqref{leader} and the linear switched system \eqref{compensator}. Under Assumptions \ref{assuniform} and \ref{ass1}, there exists a positive scalar $\varepsilon^*$, for any $0<\epsilon<\varepsilon^*$ and initial condition, such that $ \lim_{t\rightarrow\infty}(\hat{x}_i(t)-x(t))=0$, exponentially for any $x(0)\in \mathds{R}^{n}$ and $\hat{x}_i(0)\in \mathds{R}^{n}$, $i=1,\cdots,N$.
\end{thm}

By using Lemma \ref{avlemmaexp} and Lemma \ref{discrettimeexo}, we can follow the proof of Theorem \ref{thm1xe} to establish the result of Theorem \ref{theo2}. The proof is omitted. 

\section{Numerical Example}\label{numerexam}
\subsection{Example 1: Toy example}
 \begin{figure}[htbp]
 \centering
 \tikzstyle{Information} = [draw, fill=black!10, very thin,rectangle, minimum height=7.5em, minimum width=12.25em,rounded corners=.8ex]%
\subfigure[${\mathcal{G}}_{1}$]{\label{fig:a}
\begin{tikzpicture}[transform shape,scale=0.7]
  \centering%
  \path (0,0.86) node[Information](Information) {};
  \node (0) [circle, draw=red!20, fill=red!60, very thick, minimum size=7mm] {\textbf{0}};
  \node (1) [circle, right=of 0, draw=blue!20, fill=blue!60, very thick, minimum size=7mm] {\textbf{1}};
  \node (2) [circle, left=of 0, draw=blue!20, fill=blue!60, very thick, minimum size=7mm] {\textbf{2}};
  \node (3) [circle, above=of 0, draw=blue!20, fill=blue!60, very thick, minimum size=7mm] {\textbf{3}};
  \node (4) [circle, right=of 3, draw=blue!20, fill=blue!60, very thick, minimum size=7mm] {\textbf{4}};
  \node (5) [circle, left=of 3, draw=blue!20, fill=blue!60, very thick, minimum size=7mm] {\textbf{5}};
   \draw[red,thick,->] (0)--node[below]{$y_1$}(1);
   \draw[red,thick,->] (0)--node[below]{$y_2$}(2);
   \draw[red,thick,->] (0)--node[below]{$y_2$}(2);
   \draw[red,thick,->] (0)--node[left]{$y_3$}(3);
   \draw[blue,thick,<->, left] (2) edge (5);
\end{tikzpicture}
}
\quad
\subfigure[${\mathcal{G}}_{2}$]{\label{fig:b}
\begin{tikzpicture}[transform shape,scale=0.7]
  \centering%
  \path (0,0.86) node[Information](Information) {};
  \node (0) [circle, draw=red!20, fill=red!60, very thick, minimum size=7mm] {\textbf{0}};
  \node (1) [circle, right=of 0, draw=blue!20, fill=blue!60, very thick, minimum size=7mm] {\textbf{1}};
  \node (2) [circle, left=of 0, draw=blue!20, fill=blue!60, very thick, minimum size=7mm] {\textbf{2}};
  \node (3) [circle, above=of 0, draw=blue!20, fill=blue!60, very thick, minimum size=7mm] {\textbf{3}};
  \node (4) [circle, right=of 3, draw=blue!20, fill=blue!60, very thick, minimum size=7mm] {\textbf{4}};
  \node (5) [circle, left=of 3, draw=blue!20, fill=blue!60, very thick, minimum size=7mm] {\textbf{5}};
   \draw[red,thick,->] (0)--node[below]{$y_1$}(1);
   \draw[red,thick,->] (0)--node[below]{$y_2$}(2);
   \draw[red,thick,->] (0)--node[below]{$y_2$}(2);
   \draw[red,thick,->] (0)--node[left]{$y_3$}(3);
   \draw[blue,thick,->, left] (3) edge (2);
   \draw[blue,thick,->, left] (1) edge (3);
\end{tikzpicture}
}

\subfigure[${\mathcal{G}}_{3}$]{\label{fig:c}
\begin{tikzpicture}[transform shape,scale=0.7]
  \centering%
  \path (0,0.86) node[Information](Information) {};
  \node (0) [circle, draw=red!20, fill=red!60, very thick, minimum size=7mm] {\textbf{0}};
  \node (1) [circle, right=of 0, draw=blue!20, fill=blue!60, very thick, minimum size=7mm] {\textbf{1}};
  \node (2) [circle, left=of 0, draw=blue!20, fill=blue!60, very thick, minimum size=7mm] {\textbf{2}};
  \node (3) [circle, above=of 0, draw=blue!20, fill=blue!60, very thick, minimum size=7mm] {\textbf{3}};
  \node (4) [circle, right=of 3, draw=blue!20, fill=blue!60, very thick, minimum size=7mm] {\textbf{4}};
  \node (5) [circle, left=of 3, draw=blue!20, fill=blue!60, very thick, minimum size=7mm] {\textbf{5}};
   \draw[red,thick,->] (0)--node[below]{$y_1$}(1);
   \draw[red,thick,->] (0)--node[below]{$y_2$}(2);
   \draw[red,thick,->] (0)--node[below]{$y_2$}(2);
   \draw[red,thick,->] (0)--node[left]{$y_3$}(3);
   \draw[blue,thick,<->, left] (5) edge (3);
   \draw[blue,thick,->, left] (4) edge (1);
\end{tikzpicture}
}
\quad
\subfigure[${\mathcal{G}}_{4}$]{\label{fig:d}
\begin{tikzpicture}[transform shape,scale=0.7]
  \centering%
  \path (0,0.86) node[Information](Information) {};
  \node (0) [circle, draw=red!20, fill=red!60, very thick, minimum size=7mm] {\textbf{0}};
  \node (1) [circle, right=of 0, draw=blue!20, fill=blue!60, very thick, minimum size=7mm] {\textbf{1}};
  \node (2) [circle, left=of 0, draw=blue!20, fill=blue!60, very thick, minimum size=7mm] {\textbf{2}};
  \node (3) [circle, above=of 0, draw=blue!20, fill=blue!60, very thick, minimum size=7mm] {\textbf{3}};
  \node (4) [circle, right=of 3, draw=blue!20, fill=blue!60, very thick, minimum size=7mm] {\textbf{4}};
  \node (5) [circle, left=of 3, draw=blue!20, fill=blue!60, very thick, minimum size=7mm] {\textbf{5}};
   \draw[red,thick,->] (0)--node[below]{$y_1$}(1);
   \draw[red,thick,->] (0)--node[below]{$y_2$}(2);
   \draw[red,thick,->] (0)--node[below]{$y_2$}(2);
   \draw[red,thick,->] (0)--node[left]{$y_3$}(3);
   \draw[blue,thick,->, left] (3) edge (4);
\end{tikzpicture}
}
\caption{ Communication topology ${\mathcal{G}}_{\sigma(t)}$}
\label{fig1numex}
\end{figure}
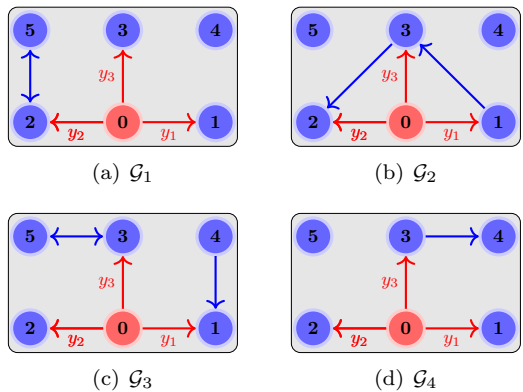
In this example, we consider a linear distributed system composed of one leader system and five followers as shown in Fig.~\ref{fig1numex}. The dynamic of the leader is in the form \eqref{leader} with
\begin{align}
A=&\left[\begin{matrix}
        0 & 2 & 0\\
       -2 & 0 & 0\\
       0 & 0 & {0.1}
       \end{matrix}
      \right],\; B=0,\;C= \begin{blockarray}{cccc}
  \begin{block}{[ccc]c}
   1& 0 & 0&C_1\\
    \cmidrule(lr){1-3}
      0 &1& 0&C_2\\
      \cmidrule(lr){1-3}
      0& 0 &1&C_3\\
      \cmidrule(lr){1-3}
      0& 0& 0&C_4\\
      \cmidrule(lr){1-3}
     0& 0& 0& C_5 \\
  \end{block}
  \end{blockarray}.
    \nonumber
\end{align}
 Next, we assume that the switching network topology $\bar{\mathcal{G}}_{\sigma\left(t\right)}$ is dictated by the following switching signal:
$$\sigma\left(t\right) =
\begin{cases}
1& \text{If $sT\leq t <\left(s+\frac{1}{4}\right)T$}\\
2& \text{If $\left(s+\frac{1}{4}\right)T\leq t <\left(s+\frac{1}{2}\right)T$}\\
3& \text{If $\left(s+\frac{1}{2}\right)T\leq t <\left(s+\frac{3}{4}\right)T$}\\
4& \text{If $\left(s+\frac{3}{4}\right)T\leq t <\left(s+1\right)T$}
\end{cases}$$
where $s=0,1,2,\cdots$. The four digraphs $\bar{\mathcal{G}}_k$, $k=1,~ 2,~ 3, ~4,$ are illustrated in Fig.~\ref{fig1numex}, where the node $0$ is associated with exosystem, and the other nodes are associated with the five followers.


\begin{figure}[htbp]
 \centering
 \epsfig{figure=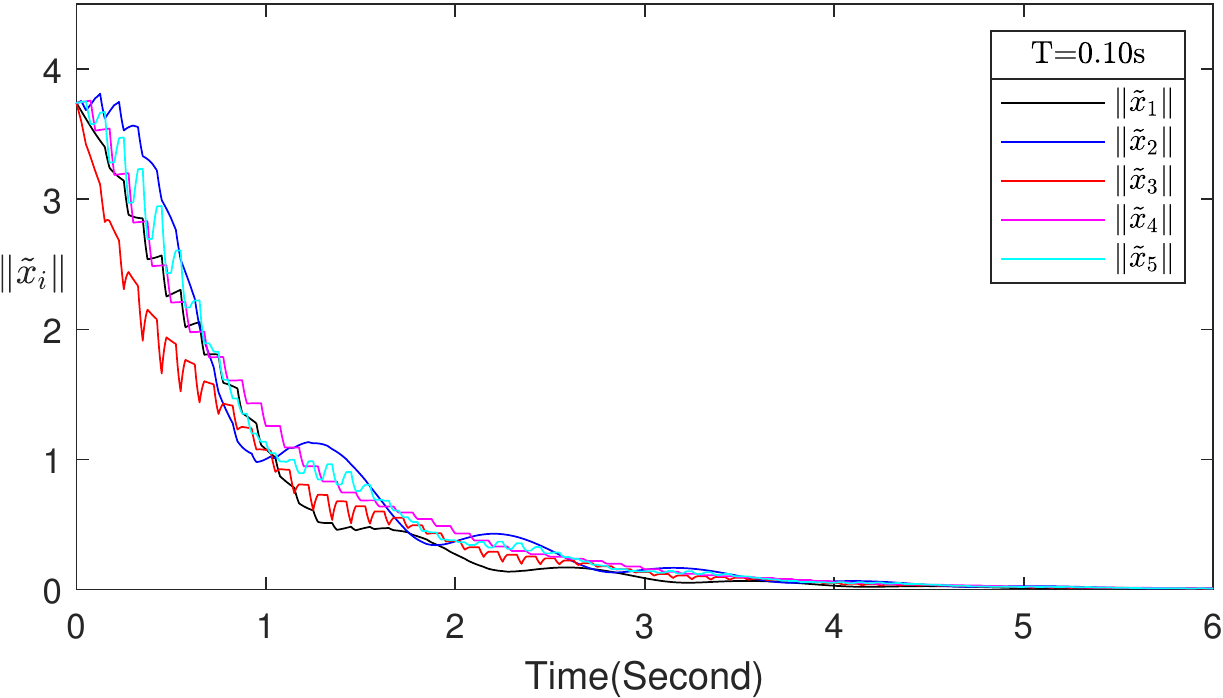,width=3.4in}
 \epsfig{figure=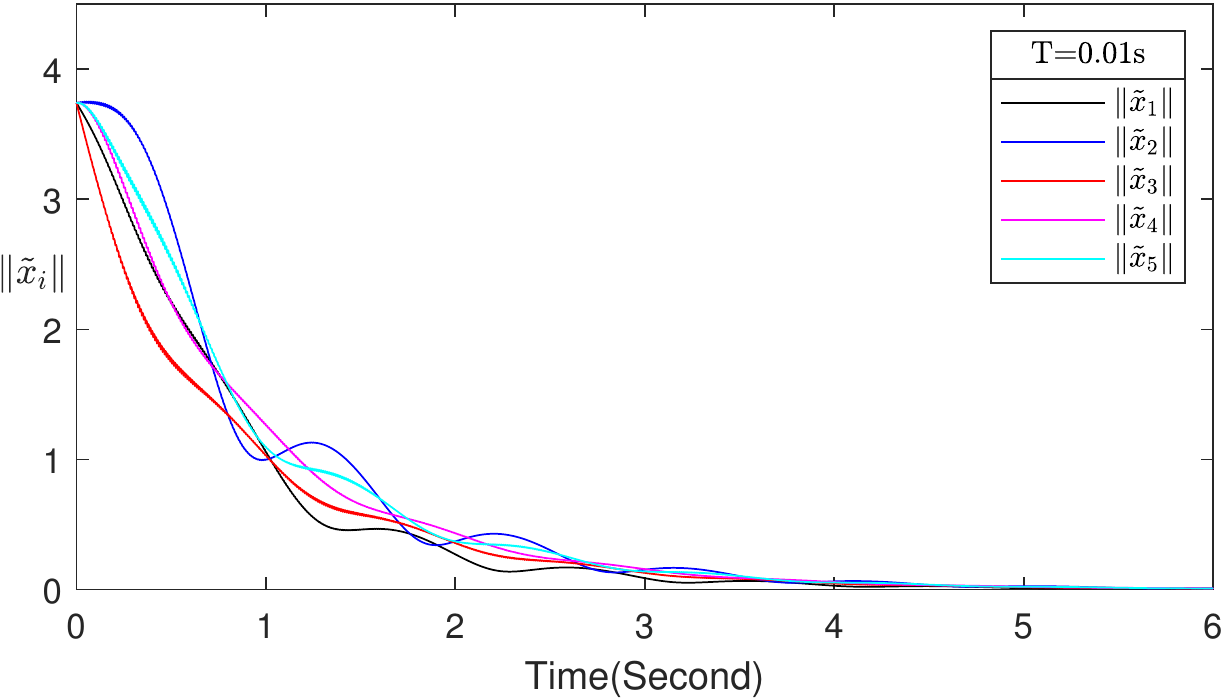,width=3.4in}
 \caption{Estimation errors of all nodes with different $T$, $i=1,\ldots,5$.}\label{figesmov}
\end{figure}

\begin{figure}[htbp]
 \centering
\epsfig{figure=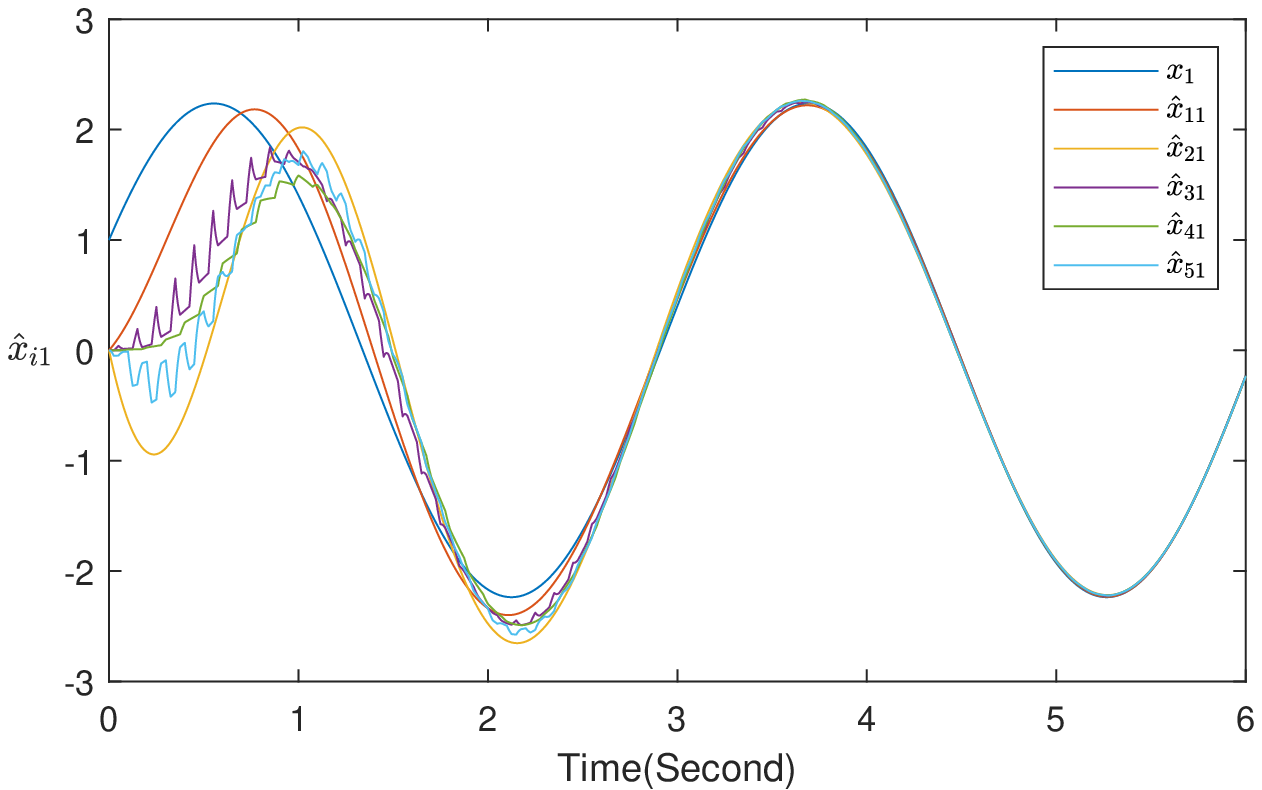,width=3.4in}
 \epsfig{figure=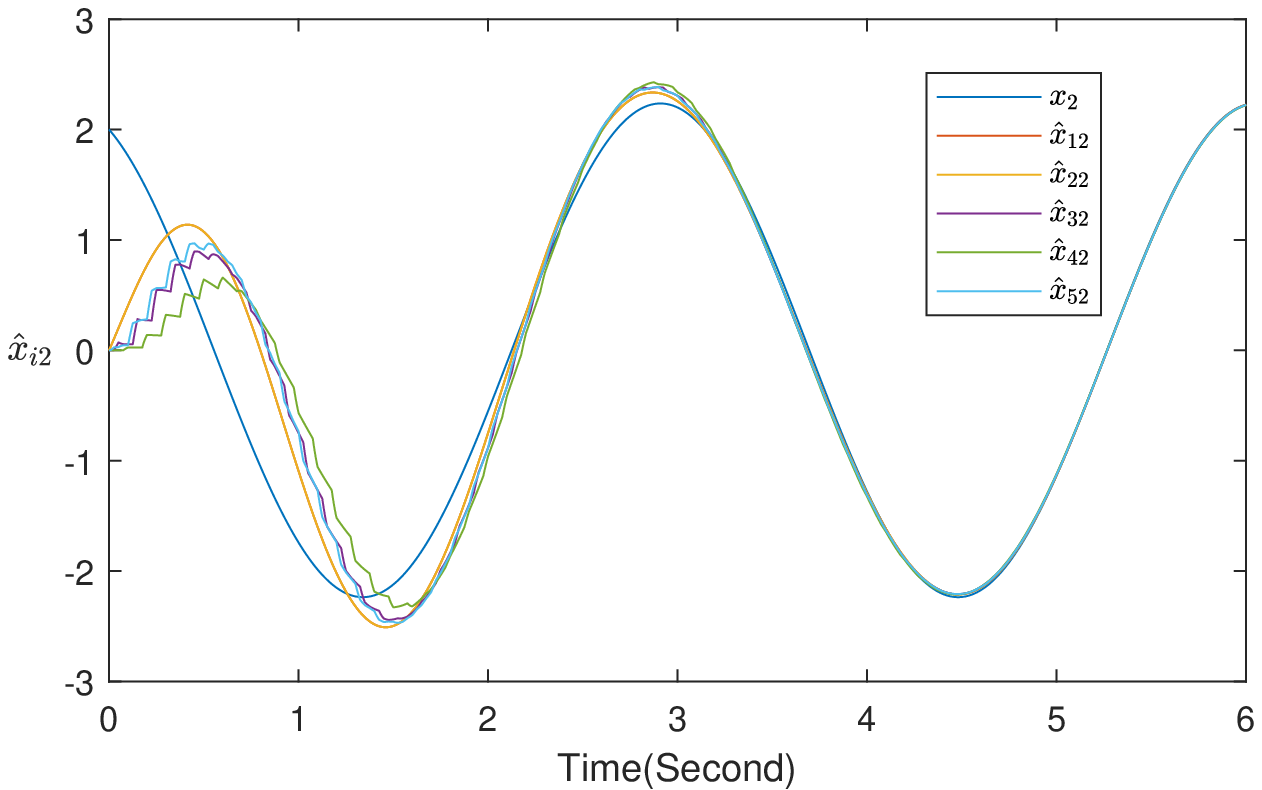,width=3.4in}
  \epsfig{figure=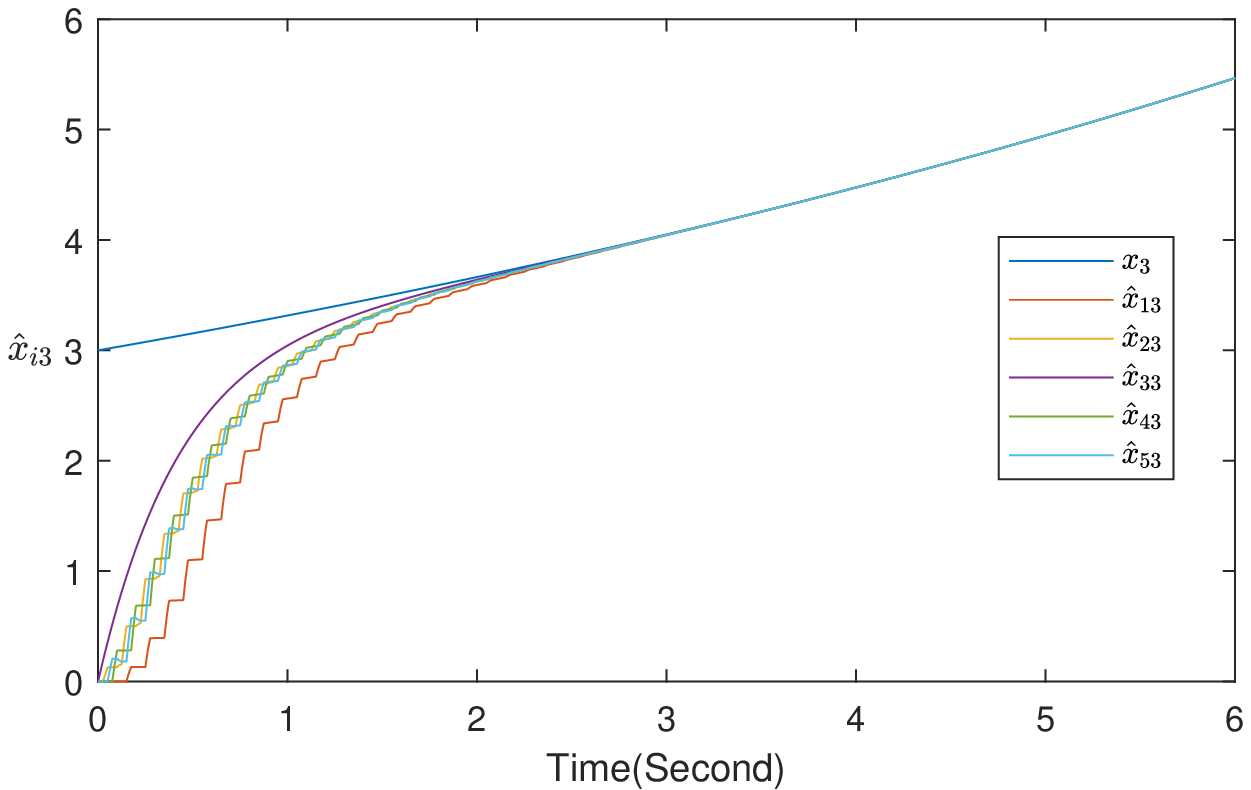,width=3.4in}
 \caption{Estimation performance of all nodes with $T=0.1$, $i=1,\ldots,5$.}\label{figesper}
\end{figure}

It can be seen that none of the pairs {\myr $(C_i, A)$} are observable, while the system to be observed is such that {\myr $(C, A)$} is observable.
Choosing the following matrix based on \eqref{decom}
\begin{align}V_1=&\left[
   \begin{matrix}
     0& 0&1\\
     0& 1&0\\
    -1&0& 0\\
   \end{matrix}
  \right],\; V_2=\left[
   \begin{matrix}
    0  & -1  & 0\\
   0 &  0 &  1\\
  -1  & 0 &  0\\
   \end{matrix}
  \right],\nonumber\\
 V_3=& \left[
   \begin{matrix}
   -1 &  0 &  0\\
   0 &  1 &  0\\
   0 &  0 &  1\\
   \end{matrix}
  \right],\;V_4=I_3,\; V_5=I_3.\nonumber \end{align}
The partition $y_1$, $y_2$ and $y_3$ of the output can be measured by agents $1$, $2$ and $3$ as shown in Fig.~\ref{fig1numex} that satisfies Assumption \ref{ass0} with
\begin{align}\mathcal{L}=\sum_{i=1}^{4}\mathcal{L}_{\rho}=\left[
    \begin{matrix}
     1 & 0 & 0 & -1& 0 \\
     0 & 2 & -1 & 0& -1 \\
     -1 & 0 & 2 & 0& -1 \\
     0 &0 & -1& 1 & 0\\
     0 &-1& 0& 0 & 1\\
    \end{matrix}
   \right].\nonumber
\end{align}
Let $\theta=10^{-3}\times\col(1,1,1,1,2)$. Hence, $\lambda_L= 1.05\times10^{-2}$ and $\lambda_l=1.3\times10^{-3}$.
Thus, we can design a control law composed of the equation \eqref{compensator} with the following parameters: $\gamma=45$,
$L_{o1}=\col(-4 -2)$, $L_{o2}=\col(-4 -2)$, $L_{o3}=-2.5$, $L_4=\col(0,0,0)$, $L_5=\col(0,0,0)$, $M_4=I_3$ and $M_5=I_3$.
Simulation is conducted with the following initial conditions:
$x(0)=\col(1,2,3)$ and $\hat{x}_i(0)=\col(0,0,0)$.

Fig.~\ref{figesmov} shows the estimation errors of all nodes over different $T$. The results demonstrate that the estimation error approaches the estimation error of the averaged estimate as the frequency of switching is increased. 
Fig.~\ref{figesper} shows the estimation performance of each agent's state over $T=0.1$. 

\subsection{Example 2: Applications to Quarter-car active automotive suspension system}

 \begin{figure}[htbp]\pagestyle{empty}
 \centering
\includegraphics[angle=270,scale=0.2]{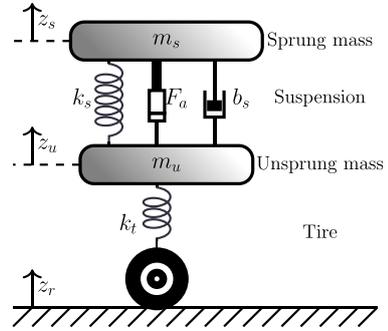}
\caption{Quarter-car active automotive suspension system}
\label{QCAASS}
\end{figure}
We borrow a modified example from \cite{yue1988alternative} and \cite{sferlazza2021state} to illustrate the design of local observers over a time-varying network.
Consider the quarter-car active automotive suspension system described in Fig.~\ref{QCAASS}, which represents the automotive system at each wheel and is monitored by two separate sensors. The system consists of a spring $k_s$, a damper $b_s$ and an active force actuator $F_a$. The sprung mass $m_s$ represents the quarter-car equivalent of the vehicle body mass. The unsprung mass $m_u$ represents the equivalent mass due to the axle and tire. The spring $k_t$ represents the vertical stiffness of the tire. The variables
$z_s$, $z_u$ and $z_r$ are the vertical displacements from static
equilibrium of the sprung mass, unsprung mass and the road, respectively.

The dynamics of this system are governed by the following system 
\begin{align}\label{qacau}
\dot{x}_c=A_cx_c+B_uF_a+B_d\dot{z}_r
\end{align}
where $F_a$ is the active force of the actuator, $\dot{z}_r$ is an input describing how the road profile enters into the system, the state vector $x_c=\col(x_{c1}, x_{c2}, x_{c3}, x_{c4})$ include the suspension deflection $x_{c1}=z_s-z_u$, the absolute velocity $x_{c2}=\dot{z}_s$ of the sprung mass $m_s$, the tire deflection $x_{c3}=z_u-z_r$ and the absolute velocity $x_{c4}=\dot{z}_u$ of the unsprung mass $m_u$. The matrices $A_c$, $B_u$ and $B_d$ are given as 
$$A_c=\left[
    \begin{array}{cccc}
     0 & 1 & 0 & -1 \\
     \frac{-k_s}{m_s} & \frac{-b_s}{m_s} & 0 & \frac{b_s}{m_s}\\
     0 & 0 & 0 & 1 \\
     \frac{k_s}{m_u} & \frac{b_s}{m_u}& \frac{-k_t}{m_u} & \frac{b_s+b_t}{-m_u} \\
    \end{array}
   \right],\begin{array}{c}
            B_u=\col\big(0,\frac{1}{m_s}, 0, \frac{-1}{m_s}\big), \\
            B_d=\col\big(0, 0, -1, \frac{b_t}{m_u}\big).
            \end{array}$$
where $m_s=240 Kg$, $m_u=36Kg$, $b_s=980 N.sec/m$, $k_s=1.6\times 10^{4}N/m$, $k_t=1.6\times 10^{5}N/m$ and $b_t=0$. 
 \begin{figure}[htbp]
 \centering
 \tikzstyle{Information} = [draw=red!20, fill=red!20, very thick,rectangle, minimum height=11.5em, minimum width=7.5em,rounded corners=.8ex]%
\subfigure[$\mathcal{G}_{1}$]{\label{figact:a}
\begin{tikzpicture}[transform shape,scale=0.7]
  \centering%
  \path (1.9,1.35) node[Information](Information) {\textbf{0}};
  \path (2,1.2) node{\includegraphics[trim=245 520 230 120,clip, scale=0.6]{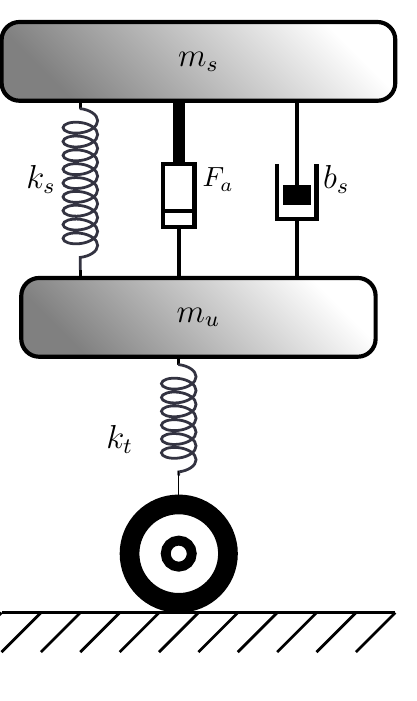}};
  \path (2.5,0.5) node{\Large \myr \textbf{0}};
  \node[circle, draw=blue!20, fill=blue!60, very thick, minimum size=7mm] at (-0.5, 0.5)  (1) {\textbf{1}};
  \node (2) [circle, above=of 1, draw=blue!20, fill=blue!60, very thick, minimum size=7mm] {\textbf{2}};
  \draw[thick,<->, left] (2) edge (1);
  \draw[red,thick,<-, left] (2) edge node[above ] {$y_2$}(0.9,2.25);
  \end{tikzpicture}}
\quad
\subfigure[$\mathcal{G}_{2}$]{\label{figact:b}
\begin{tikzpicture}[transform shape,scale=0.7]
  \centering%
  \path (1.9,1.35) node[Information](Information) {};
   \path (2.5,0.5) node{\Large \myr \textbf{0}};
  \path (2,1.2)node{\includegraphics[trim=245 520 230 120,clip, scale=0.6]{figures/active.pdf}};
  \node[circle, draw=blue!20, fill=blue!60, very thick, minimum size=7mm] at (-0.5, 0.5)  (1) {\textbf{1}};
  \node (2) [circle, above=of 1, draw=blue!20, fill=blue!60, very thick, minimum size=7mm] {\textbf{2}};
  \draw[red,thick,<-, left] (2) edge node[above ] {$y_2$}(0.9,2.25);
\end{tikzpicture}}
\caption{ Communication topology and measurements $\mathcal{G}_{\sigma\left(t\right)}$}
\label{fig1numexac}
\end{figure}

The time-varying networks are shown in Fig.~\ref{fig1numexac}. The network structure is driven by the
switching signal:
$$\sigma\left(t\right) =
\begin{cases}
1& \text{If $sT\leq t <\left(s+\frac{1}{2}\right)T$}\\
2& \text{If $\left(s+\frac{1}{2}\right)T\leq t <\left(s+1\right)T$}\\
\end{cases}$$
where $s=0,1,2,\cdots$. 
We assume that the road profile can be viewed as a disturbance described by the following equation
\begin{subequations}\label{exosytdis}\begin{align}
\dot{v}=Sv,\\
\dot{z}_r=Fv,
\end{align}\end{subequations}
where $v\in \mathds{R}^4$ is the state of the exosystem, $F=\left[\begin{smallmatrix}2&-1&-1&2\end{smallmatrix}\right]$ and $S=\diag(2a,a)$ are the output matrix and system matrix of the exosystem with $a=\left[\begin{smallmatrix}
  0&1\\
  -1&0\\
\end{smallmatrix}\right]$. We define $x=\col(x_c,v)$ and rewrite\eqref{exosytdis} and \eqref{qacau} into the following compact form:
\begin{align}
 \dot{x}=Ax+BF_a
\end{align}
where $A=\left[\begin{smallmatrix}A_c&F\\
 0_{4\times 4}& S\end{smallmatrix}\right]$ and $B=\left[\begin{smallmatrix}B_u\\0_{4}\end{smallmatrix}\right]$. It is also assumed that the active force actuator is $F_a=0$. The variable $y_2=C_2x$ is measured by agent $2$ with $C_2^T=\col(1,0,0,0,0_4)$. Agent $1$ does not measure anything from the quarter-car active automotive suspension system. 
For this configuration, we get that $rank(\mathcal{O}_1) = 0$ and $rank(\mathcal{O}_2) = 8$. Furthermore, we confirm that {\myr $(C, A)$} is observable with $C=\col(C_1,C_2)$. Using the decomposition \eqref{decom}, we compute the transformed matrices and choose $\gamma=25$ and the vector $L_{o2}$ such that the matrix $(A_{o2}+L_{o2}C_{o2})$ at $\{-2;-4;-6;-8;-3;-9;-5;-7\}$.
The initial conditions for the simulation are generated randomly. 
 \begin{figure}[htbp]
 \centering
  \epsfig{figure=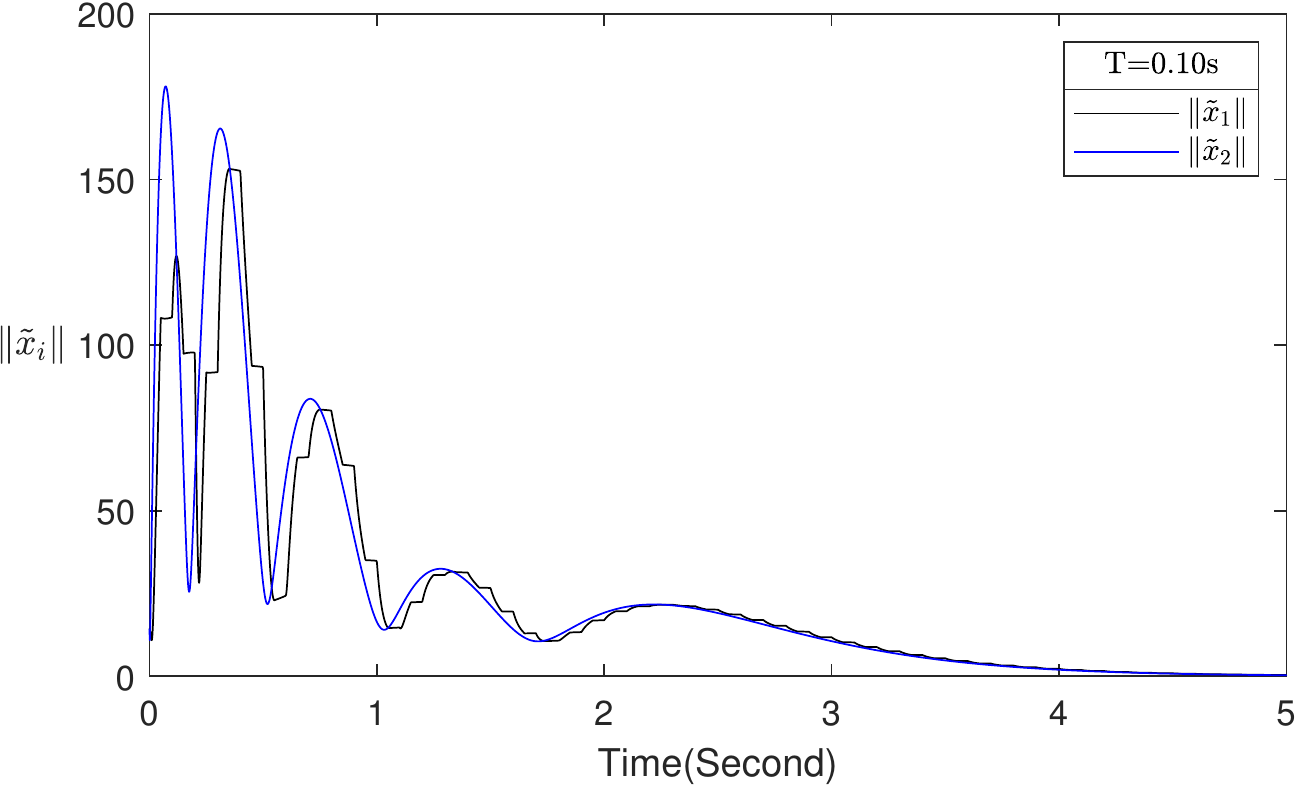,width=3.4in} 
 \epsfig{figure=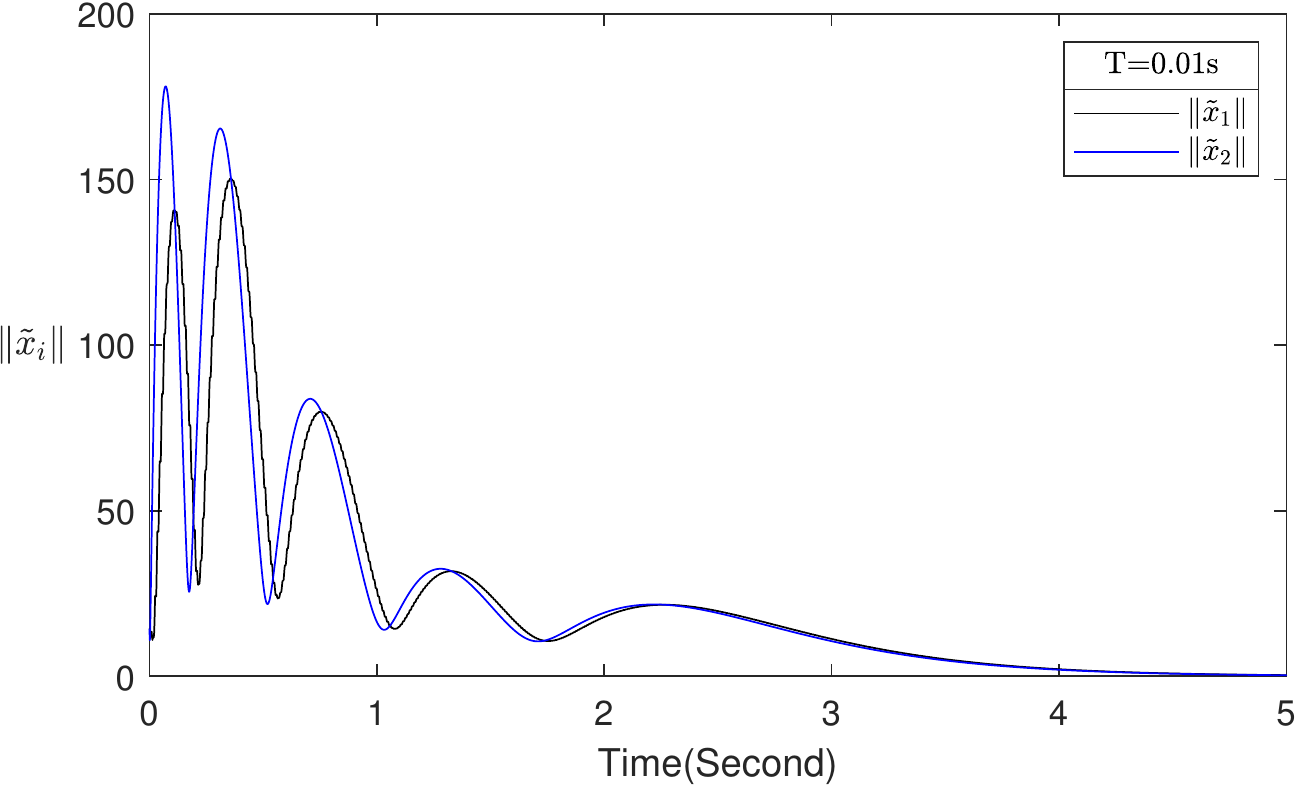,width=3.4in} 
 \caption{Estimation errors of all nodes with different $T$.}\label{figessec}
 \end{figure}
 Fig.~\ref{figessec} shows the estimation errors of all
agents. It can be observed that agent $2^{th}$ can reconstruct the state of \eqref{qacau}. It also demonstrates that the distributed observer with fast-switching communication enables the estimation for the remote agent $1^{th}$ which can approximate the averaged system.

\subsection{Example 3: Applications to the formation of unicycle-type vehicles}

We now consider formation control of $5$ flying unicycle-type vehicles with the $i^{th}$ agent's dynamics given by:
\begin{align}\label{flying}
  \dot{p}_{xi}=&\nu_i \cos (\theta_i),\; \dot{p}_{yi}=\nu_i \sin (\theta_i),\; \dot{\theta}_i=\omega_i
\end{align}
where $p_i=\col(p_{xi},p_{yi})\in \mathds{R}^2$ is the position along the $X-$ and $Y-$axis, $\theta_i$ is the attitude, $\nu_i$ is the translational speed, and $\omega_i$ is the angular velocity of the $i^{th}$ vehicle, respectively. 
 \begin{figure}[htbp]
   \centering%
  \includegraphics[trim=142 586 112 126,clip, width=3.3in]{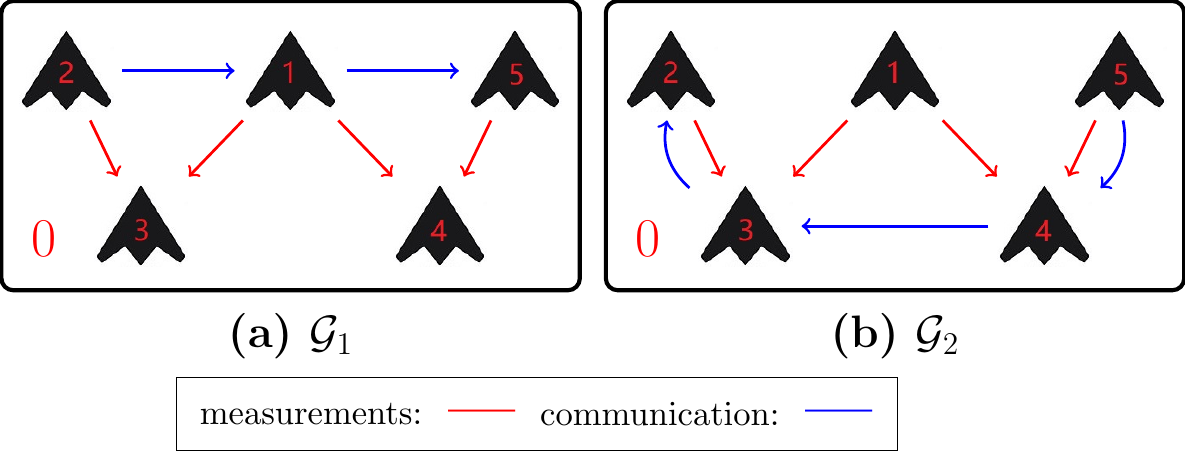}
  \caption{Communication networks $\mathcal{G}_{\sigma\left(t\right)}$ of flying vehicles}\label{flyingnet} 
  \end{figure}
  
 \begin{figure}[htbp]
   \centering%
  \includegraphics[width=3.4in]{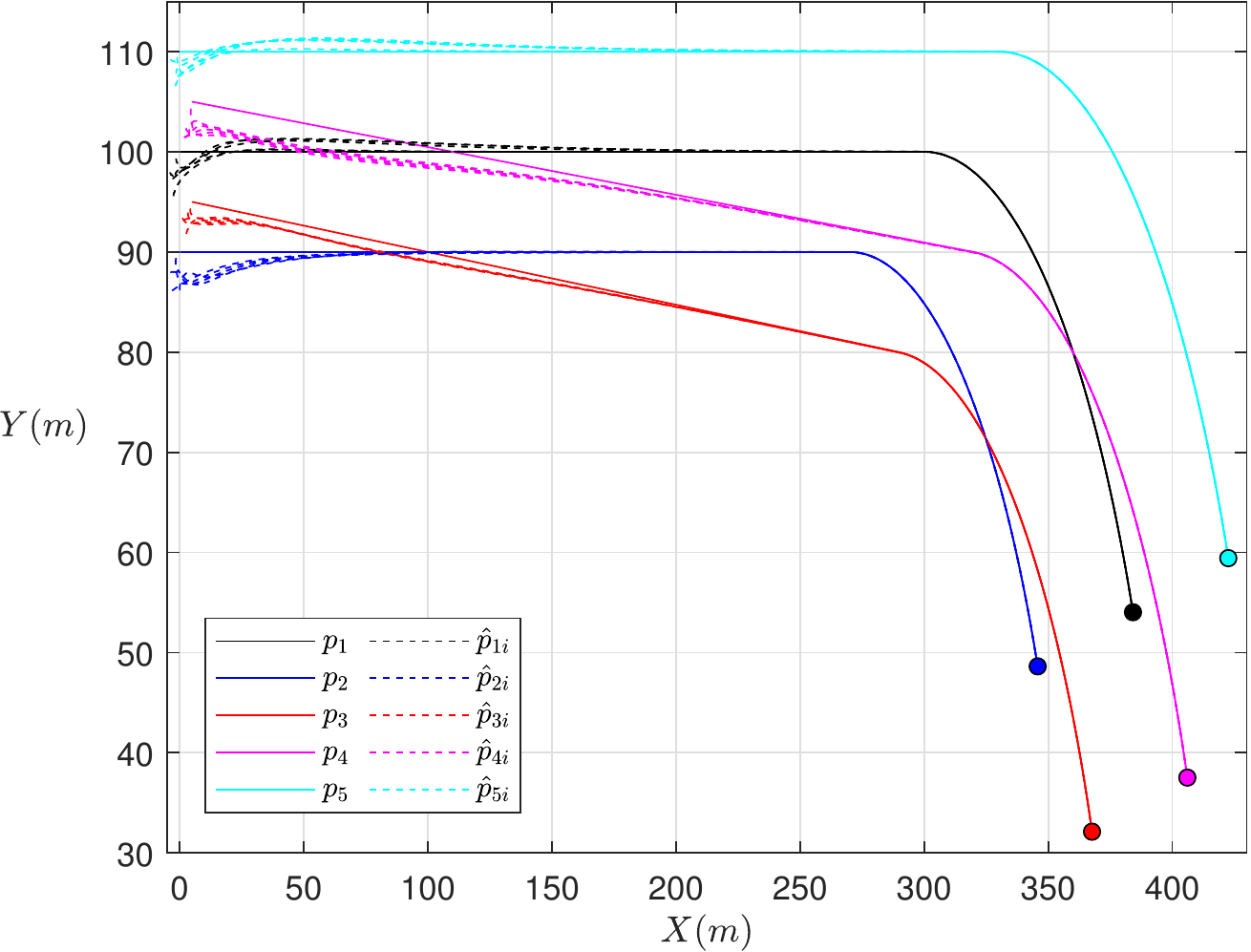}
  \caption{Position and Estimated position of all vehicles.}\label{figflyingPo}
 \end{figure}
  
 \begin{figure}[htbp]
  \epsfig{figure=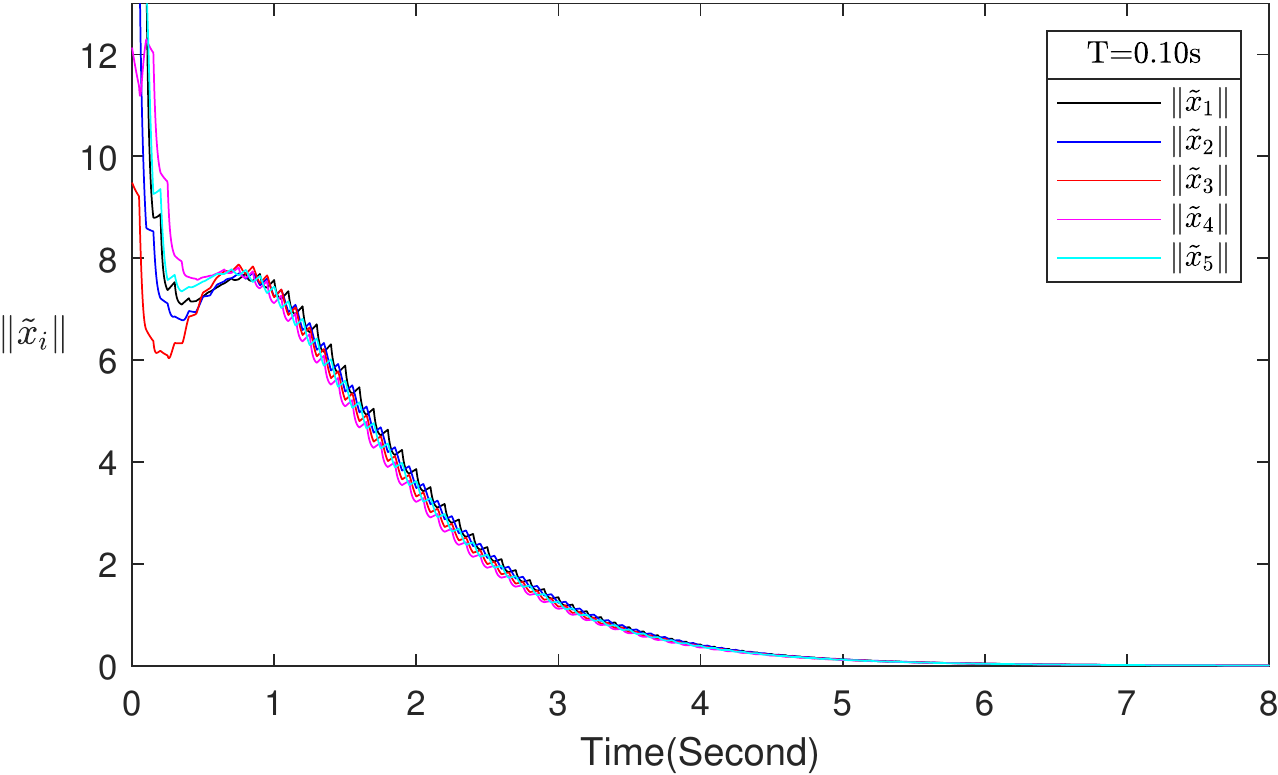,width=3.4in}
  \epsfig{figure=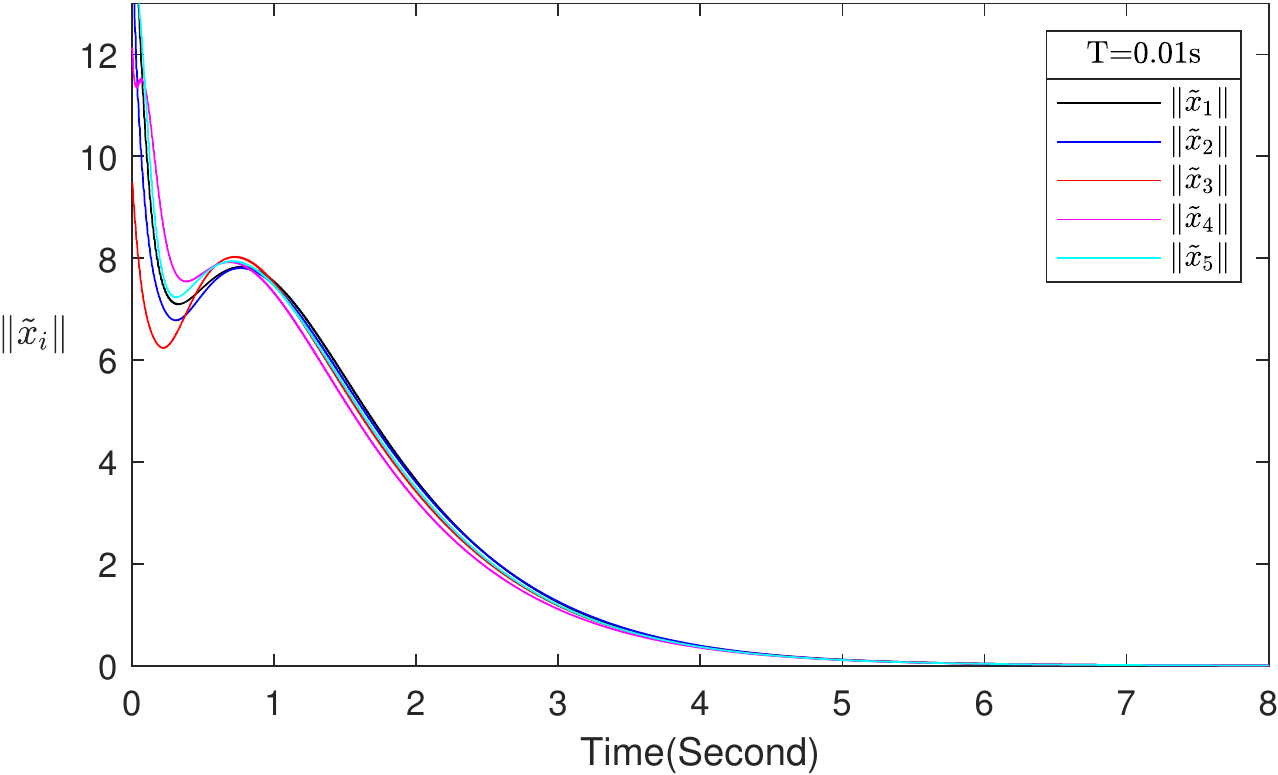,width=3.4in}
 \caption{Estimation errors of all vehicles with different $T$.}\label{figflyinges}
 \end{figure}
 Same as in \cite{kim2019completely}, we assume that the formation flight is conducted in unforced conditions for which $\nu_i$ and $\omega_i$ are constants. Let $v_i=\col(v_{xi}, v_{yi})$ represent the velocity of the $i^{th}$ vehicle along the $X-$ and $Y-$axis. Using the dynamics of \eqref{flying}, we can establish that the dynamics of each vehicle are governed by the following two systems:
\begin{align*}
\left[\begin{matrix}
\dot{\theta}_i\\
\dot{\omega}_i\\
\end{matrix}\right]=&\left[\begin{matrix}
0&1\\
0&0\\
\end{matrix}\right]\left[\begin{matrix}
{\theta}_i\\
{\omega}_i
\end{matrix}\right],\\
\left[\begin{matrix}
\dot{p}_i\\
\dot{v}_i\\
\end{matrix}\right]=&\left[\begin{matrix}
0_{2\times 2}&I_2\\
0_{2\times 2}& \omega_i a\\
\end{matrix}\right]\left[\begin{matrix}
{p}_i\\
{v}_i
\end{matrix}\right],
\end{align*}
where $a=\left[\begin{smallmatrix}
  0&1\\
  -1&0\\
\end{smallmatrix}\right]$. We assume that each vehicle has knowledge of its altitude and angular velocity $\omega_i$. 
Each vehicle has only limited measurement ability . For example, as shown in Fig.~\ref{flyingnet}, each vehicle can access its own position and velocity. However, only agents $3$ and $4$ can measure their relative position to other agents. 
 The communication network structure and the local measurements for the network of vehicles are shown in Fig.~\ref{flyingnet}. We define $x=\col(p_1,v_1,\cdots,p_5,v_5)\in \mathds{R}^{20}$. The network of multi-vehicles can be written in the form of system \eqref{leader} with
\begin{align*}
A=& \diag(A_1,A_2,A_3,A_4,A_5),\;
C_1=\left[\begin{matrix}1&0&0&0&0\end{matrix}\right] \otimes S,\\
C_2=&\left[\begin{matrix}0&1&0&0&0\end{matrix}\right]\otimes S,\;
C_3=\left[\begin{matrix}1&1&1&0&0\end{matrix}\right]\otimes S,\\
C_4=&\left[\begin{matrix}1&0&0&1&1\end{matrix}\right]\otimes S,\;
C_5=\left[\begin{matrix}0&0&0&0&1\end{matrix}\right]\otimes S,
\end{align*}
where $A_i=\left[\begin{smallmatrix}
0_{2\times 2}&I_2\\
0_{2\times 2}& \omega_i\times a\\
\end{smallmatrix}\right] $ and $S=\left[\begin{matrix}I_2&0_{2\times 2}\end{matrix}\right]$ are the system and output matrix of $i^{th}$ vehicle, $C_i$ reflects the measurement ability of $i^{th}$ vehicles in terms of other vehicles, $A$ is system matrix of the system to be observed. We can also verify that $rank(\mathcal{O}_1) = 4$, $rank(\mathcal{O}_2) = 4$, $rank(\mathcal{O}_3) = 4$, $rank(\mathcal{O}_4) = 4$ and $rank(\mathcal{O}_5) = 4$. As a result, no vehicle can observe the global information despite the fact that the pair {\myr $(C, A)$} is observable with $C=\col(C_1, C_2, C_3, C_4, C_5)$. Using the proposed observer, each vehicle can establish distributed observer to estimate the global information $x$ over the time-varying networks shown in Fig.~\ref{flyingnet} with switching signal:
$$\sigma\left(t\right) =
\begin{cases}
1& \text{If $sT\leq t <\left(s+\frac{1}{2}\right)T$}\\
2& \text{If $\left(s+\frac{1}{2}\right)T\leq t <\left(s+1\right)T$}\\
\end{cases}$$
where $s=0,1,2,\cdots$. 

As above, the tuning parameters are chosen using the Kalman decomposition \eqref{decom} for different values of $\omega_i$. The simulation starts with angular velocity $\omega_i=0$. At the time $t=6s$, the angular velocities of the vehicles are changed from $\omega_i=0$ to $\omega_i=0.5$. Fig.~\ref{figflyingPo} shows the estimated positions and positions of all followers. Fig.~\ref{figflyinges} shows the estimation errors of all followers. We compute the transformed matrices and choose $\gamma=25$ and the vector $L_{0i}$ such that the matrix $(A_{oi}+L_{oi}C_{oi})$ has eigenvalues at $\{-1,-2,-3,-4\}$.

\subsection{Example 4: Applications to position estimation in vehicular platoons of multiple connected and automated vehicles (CAV)}

 \begin{figure}[htbp]
   \centering%
  \includegraphics[trim=162 535 163 128,clip, width=3.3in]{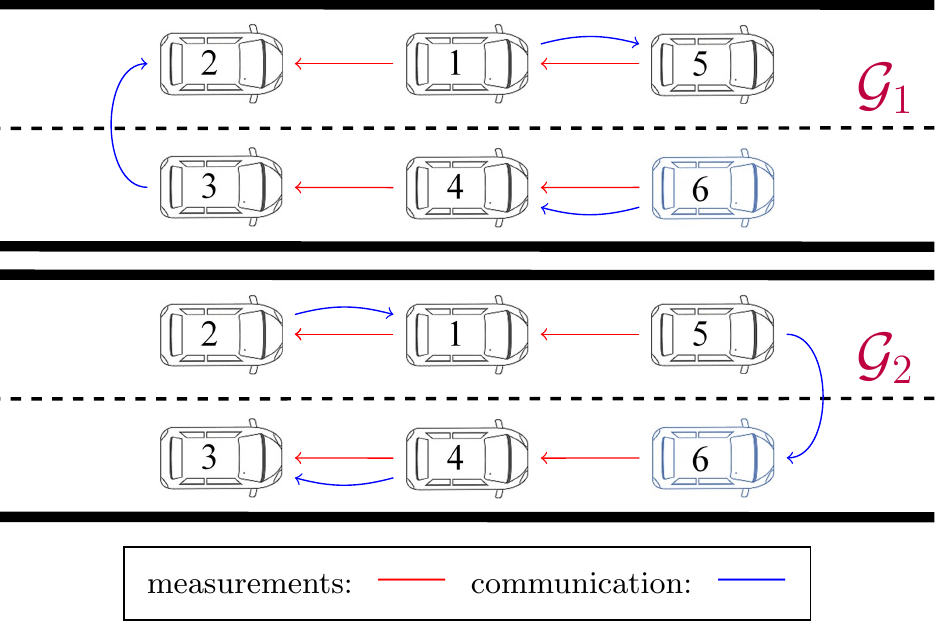}
  \caption{Connected autonomous vehicles' communication networks $\mathcal{G}_{\sigma\left(t\right)}$}\label{CAVnet} 
  \end{figure}  
We now apply the distributed observer to the vehicular platoons control of multiple connected and automated vehicles with dynamics given by
\begin{align*}
\left[\begin{matrix}
\dot{p}_i\\
\dot{v}_i\\
\end{matrix}\right]=&\left[\begin{matrix}
0_{2\times 2}&I_2\\
0_{2\times 2}& 0_{2\times 2}\\
\end{matrix}\right]\left[\begin{matrix}
{p}_i\\
{v}_i
\end{matrix}\right]+\left[\begin{matrix}
0_{2\times 1}\\
u_{i}\\
\end{matrix}\right],
\end{align*}
where $p_i=\col(p_{xi},p_{yi})\in \mathds{R}^2$ and $v=\col(v_{xi},v_{yi})\in \mathds{R}^2$ are the position and velocity of the $i^{th}$ CAV along the $X-$ and $Y-$axis, respectively; $u_i\in \mathds{R}^2$ is the control input of the $i^{th}$ CAV along the $X-$ and $Y-$axis.
  \begin{figure}[htbp]
  \centering
  \epsfig{figure=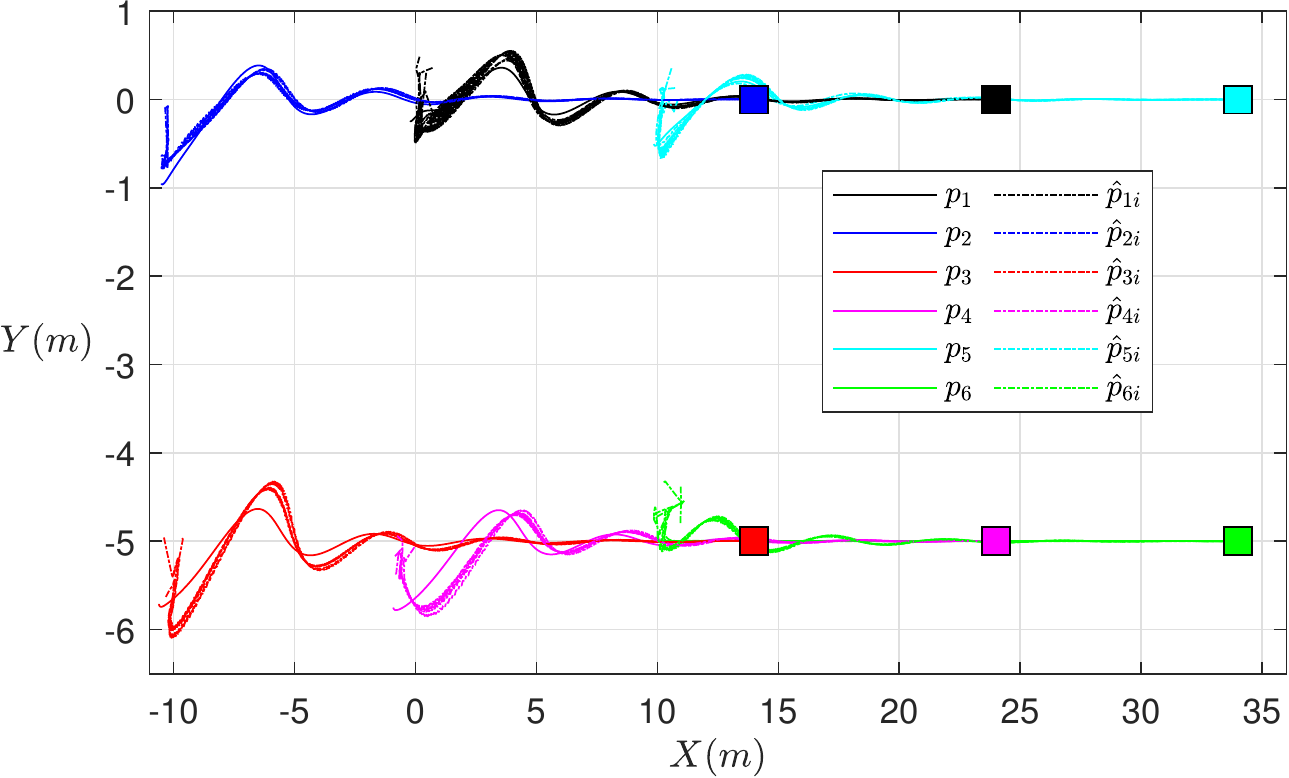,width=3.4in}
  \caption{Position and Estimated position of all vehicles.}\label{transingPo}
 \end{figure}
 
  \begin{figure}[htbp]
  \epsfig{figure=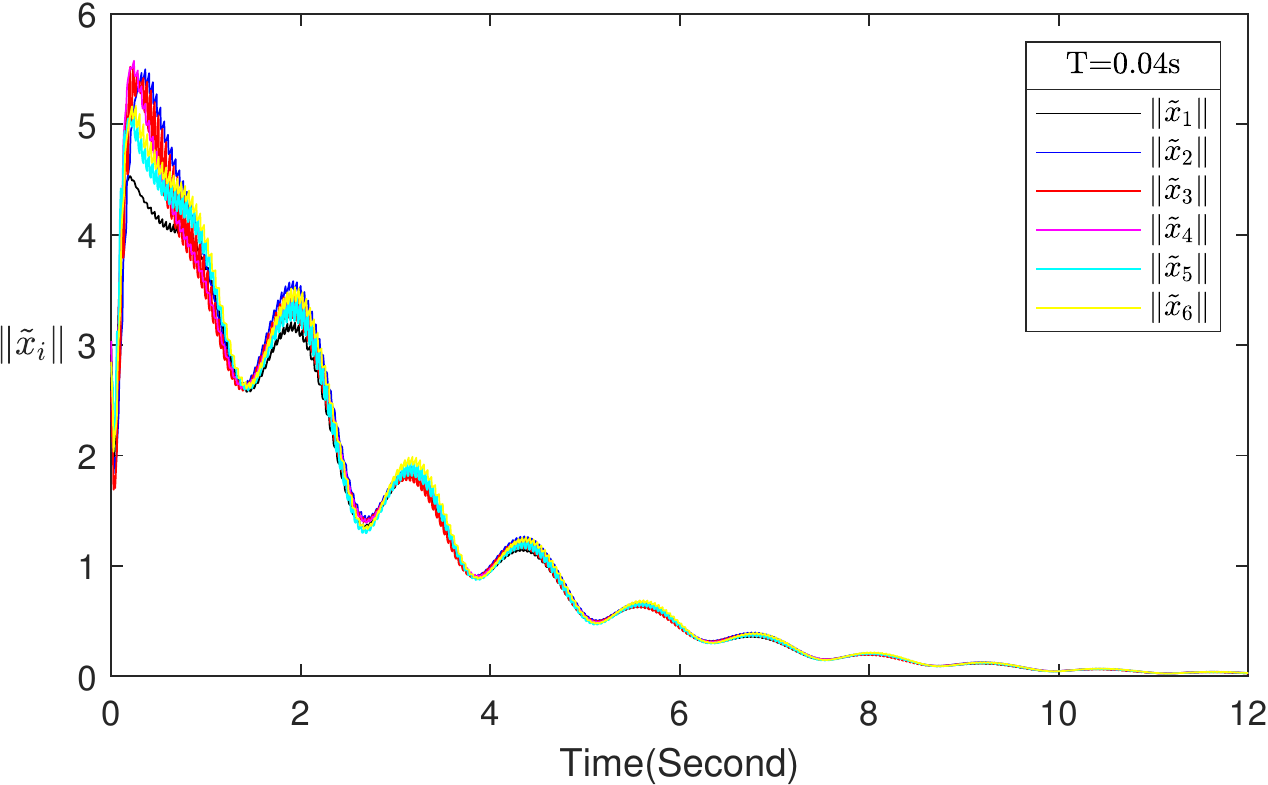,width=3.4in}    
  \epsfig{figure=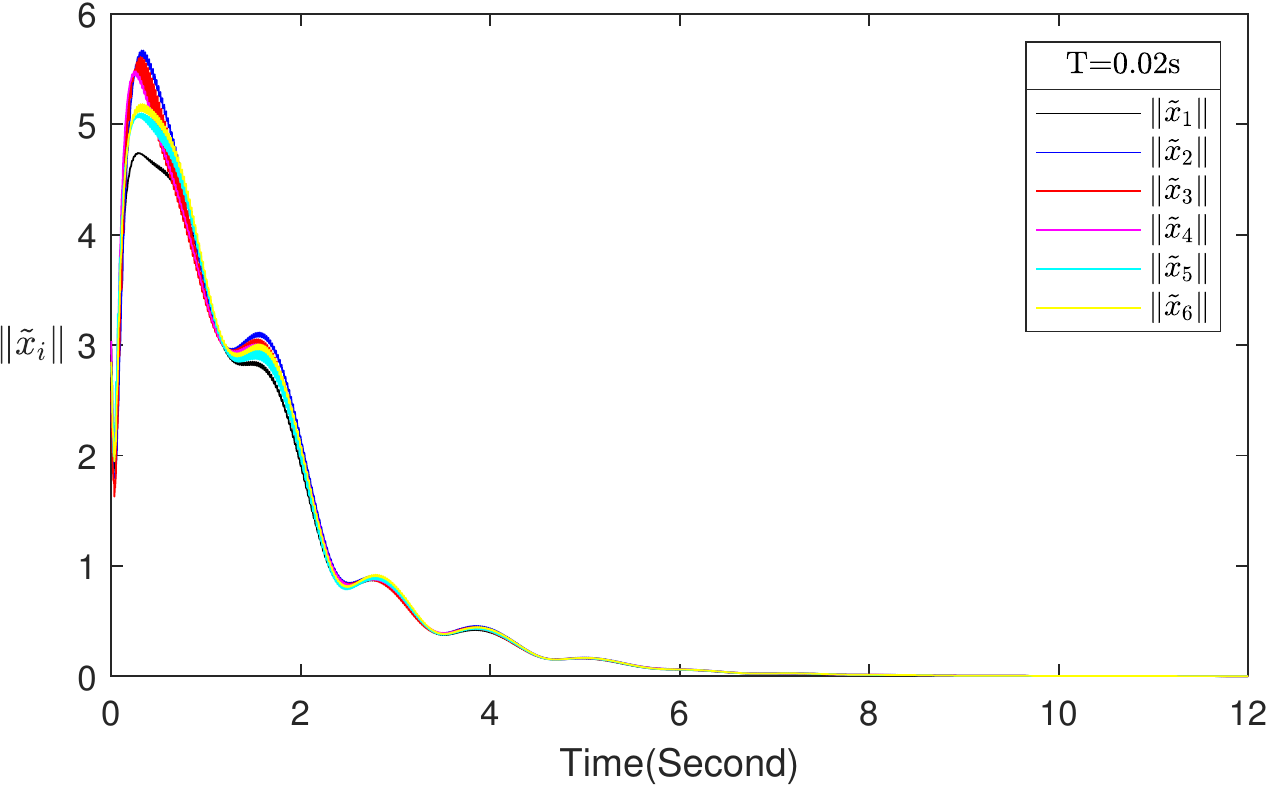,width=3.4in}
 \caption{Estimation errors of all vehicles with different $T$.}\label{transinges}
 \end{figure}
 
Each vehicle has limited measurement ability. As shown in Fig.~\ref{CAVnet}, each vehicle can access its own position and velocity, but only a vehicle following another one can measure its relative position to the vehicle ahead.
The network of all vehicles, their local measurements and communication network structure is shown in Fig.~\ref{flyingnet}. In this case, the state of the network is denoted by $x=\col(p_1,v_1,\cdots,p_6,v_6)\in \mathds{R}^{24}$ with dynamics \eqref{leader} with system matrix
$A=I_5\otimes \left[\begin{smallmatrix}
0_{2\times 2}&I_2\\
0_{2\times 2}& 0_{2\times 2}\\
\end{smallmatrix}\right]$,
and with local measurement matrices
$C_1=\left[\begin{matrix}1&0&0&0&1&0\end{matrix}\right] \otimes S$,
$C_2=\left[\begin{matrix}1&1&0&0&0&0\end{matrix}\right]\otimes S$,
$C_3=\left[\begin{matrix}0&0&1&1&0&0\end{matrix}\right]\otimes S$,
$C_4=\left[\begin{matrix}0&0&0&1&0&1\end{matrix}\right]\otimes S$,
$C_5=\left[\begin{matrix}0&0&0&0&1&0\end{matrix}\right]\otimes S$ and $C_6=\left[\begin{matrix}0&0&0&0&0&1\end{matrix}\right]\otimes S$.
where $S=\left[\begin{matrix}I_2&0_{2\times 2}\end{matrix}\right]$. Each $C_i$ reflects the measurement ability of $i^{th}$ vehicle.

We can also verify that $rank(\mathcal{O}_1) = 4$, $rank(\mathcal{O}_2) = 4$, $rank(\mathcal{O}_3) = 4$, $rank(\mathcal{O}_4) = 4$, $rank(\mathcal{O}_5) = 4$, $rank(\mathcal{O}_6) = 4$ and $rank(\mathcal{O}) = 24$ with $C=\col(C_1, C_2, C_3, C_4, C_5, C_6)$. As a result, no vehicle can reconstruct the state of other vehicles without coordination. 

Suppose that vehicle $6$ leads the vehicular network and that the following reference signal describes its position:
 $$p_6(t)=\col(v_{6x}t+p_{6x},v_{6y}t+p_{6y}).$$
 The control objective is to control each vehicle such that 
 $$ \lim_{t\rightarrow\infty}\left(p_i(t)-p_6(t)\right)=p_{i6}$$
 where $p_{i6}\in \mathds{R}^2$ are the relative positions between $i^{th}$ vehicle and vehicle $6$. The control
 input is chosen as
 $$u_i=K(x_i-\hat{x}_{6i}-\col(p_{i6},0,0))$$
 where $K^T=\col(-8,-4)\otimes I_2$, $x_i=\col(p_i,v_i)$ and $\hat{x}_{i6}=\col(\hat{p}_{i6},\hat{v}_{i6})$ is the vector of estimates of the states variables of vehicle 6 for the $i^{th}$ vehicle. We compute the transformed matrices using the Kalman decomposition \eqref{decom}, and choose $\gamma=250$ and the vector $L_{oi}$ such that the matrix $(A_{oi}+L_{oi}C_{oi})$ at $\{-1,-2,-3,-4\}$.
Fig.~\ref{transingPo} shows the estimated positions and positions of all followers. Fig.~\ref{transinges} shows the estimation errors of all followers.
Each vehicle has limited measurement ability and cannot measure the position of all the cars ahead. As we can see from Fig.~\ref{transingPo}, the use of the time-varying communication networks enables the reconstruction of the state of all vehicles from all vehicles in the network Fig.~\ref{CAVnet}. 

\section{Conclusion}\label{conlu}
A distributed observer for a general linear time-invariant system is proposed to estimate the state of the system over time-varying jointly observable networks. An averaging approach is proposed to facilitate the analysis of the distributed observer scenarios that meet an uniformly connected on average assumption. This assumption is met by a wide range of possible switching protocols, including periodic switching, Markovian switching and Cox process switching. It is shown that the distributed observer provides exponential convergence of the state estimates for each agent. 
%
%
%
%
%
%
The study provides two exponential stability results for two classes of switched systems, providing valuable tools for future research in related distributed state estimation fields.

\noindent
\bibliographystyle{ifacconf}
\bibliography{myref}\end{sloppypar}
\end{document}